\documentclass[twocolumn,showpacs]{revtex4-1}

\usepackage{newlfont}
\usepackage[below]{placeins}
\usepackage{flafter}
\usepackage{amsfonts}
\usepackage{amsmath}
\usepackage{amssymb}
\usepackage[american]{babel}
\usepackage{graphicx}
\usepackage{subfigure}
\usepackage{multirow}
\usepackage{url}
\usepackage{color}
\usepackage{graphicx}
\usepackage{mathtools}

\allowdisplaybreaks

\newcommand{\ket}[1]{\big| #1 \big\rangle}
\newcommand{\bra}[1]{\big\langle #1 \big|}
\newcommand{\braket}[2]{\big\langle #1 \big| #2 \big\rangle}                 

\newtheorem{theorem}{Theorem}[section]
\newtheorem{lemma}[theorem]{Lemma}
\newtheorem{proposition}[theorem]{Proposition}

\newtheorem{definition}[theorem]{Definition}

\newenvironment{proof}[1][Proof]{\begin{trivlist}
\item[\hskip \labelsep {\bfseries #1}]}{\end{trivlist}}

\newcommand{\qed}{\nobreak \ifvmode \relax \else
      \ifdim\lastskip<1.5em \hskip-\lastskip
      \hskip1.5em plus0em minus0.5em \fi \nobreak
      \vrule height0.75em width0.5em depth0.25em\fi}

\begin{document}


\title{Staggered Quantum Walks on Graphs}
\author{Renato Portugal}
\affiliation{ National Laboratory of Scientific Computing (LNCC) \\
Petr\'{o}polis, RJ,  25651-075, Brazil\\
}

\date{\today}

\begin{abstract}
The staggered quantum walk model allows to establish an unprecedented connection between discrete-time quantum walks and graph theory. We call attention to the fact that a large subclass of the coined model is included in Szegedy's model, which in its turn is entirely included in the staggered model. In order to compare those three quantum walk models, we put them in the staggered formalism and we show that the Szegedy and coined models are defined on a special subclass of graphs. This inclusion scheme is also true when the searching framework is added. We use graph theory to characterize which staggered quantum walks can be reduced to the Szegedy or coined quantum walk model. We analyze a staggered-based search that cannot be included in Szegedy's model and we show numerically that this search is more efficient than a random-walk-based search.
\end{abstract}

\maketitle

\section{Introduction}

Coined QWs on graphs were defined in Ref.~\cite{Aharonov:2000}, have been extensively analyzed in literature~\cite{Ven12,Kon08,Kendon:2007,Portugal:Book,Manouchehri2014}, and were used to develop new quantum algorithms, such as, for searching marked vertices on graphs~\cite{Shenvi:2003,Ambainis:2005}. Despite the success of coined QWs, Szegedy~\cite{Szegedy:2004} proposed a new discrete-time QW model without coins on bipartite graphs, which was also used to build new quantum algorithms, for instance, for searching triangles in graphs~\cite{mss07} and for searching webpages in complex networks~\cite{PMCM13}. Generalizations of Szegedy's model was recently proposed in some papers~\cite{HKSS14,MOS16}.

The staggered quantum walk (SQW) model was defined in Ref.~\cite{PSFG15}, which has shown that the entire Sgezedy's model~\cite{Szegedy:2004}, including its searching framework, is contained in the SQW model. Ref.~\cite{Por16} has shown that many coined QWs (DTQWs) can be cast into Szegedy's model, including flip-flop coined QWs employing the Grover or Hadamard coins and the coined QWs using the abstract-search-algorithm scheme~\cite{Ambainis:2005}. Ref.~\cite{Por16} has also shown that if the DTQW on a graph $\Gamma$ is included in Szegedy's model, then the DTQW can be seen as a coinless QW on an enlarged graph $\Gamma'$. The coin space, which is internal in the DTQW, becomes explicit in the equivalent SQW on $\Gamma'$.

The name of the SQW model comes from the staggered fermion formalism~\cite{KS75,Sus77,STW81}, which was proposed to solve some technical difficulties when dealing with fermionic fields in the context of quantum field theories. This formalism is useful to put fermionic fields on a hypercubic lattice to be addressed in the context of the lattice field theory~\cite{Kog79}. Key ideas of the staggered fermion formalism were used in Ref.~\cite{Meyer96} to propose a nontrivial one-dimensional quantum cellular automata avoiding the no-go lemma~\cite{Mey96b}. Similar ideas were used in Refs.~\cite{Patel05,Patel:2010,Patel:2010b} to propose coinless QWs on hypercubic lattices, the one-dimensional case of which includes the one proposed in Ref.~\cite{Meyer96} as a particular case. Refs.~\cite{HKS05} have shown that the one-dimensional case can be included in the DTQW model. The higher dimensional versions, in the way presented in Ref.~\cite{Patel05}, use nonlocal unitary operators violating the principle that walkers must jump only to neighboring sites. The escape for this problem, when the dimension is greater than one, is to assume that the graph on which the QW has been defined is not the hypercubic lattice by adding edges connecting some diagonally-opposed sites as discussed in Ref.~\cite{PSFG15}. 

The basic ideas employed in Refs.~\cite{Meyer96,Patel05} can be summarized in two points. First, they convert the internal spin or quirality state of the particle into extra degrees of freedom by adding new vertices to the lattice, which becomes a larger lattice (noticeable only in the finite case with some boundary conditions). Second, they use two unitary operators with repeated alternating action. Can those ideas be applied for the coined model on a generic graph? Ref.~\cite{Por16} has addressed this problem for a subclass of flip-flop coined models characterized by coins with $(+1)$-eigenvectors obeying special orthogonality properties, called orthogonal reflections. For $d$-dimensional coins, the original graph of the coined model must be enlarged by replacing each vertex with a $d$-\textit{clique}~\footnote{In this work, we employ many technical terms of graph theory, which are in italics to indicate that they are in the glossary in Appendix~\ref{appendixA}.} and by using two alternating unitary operators described by the staggered QW model. The result is a coinless QW on the enlarged graph equivalent to the coined model on the original graph.

The staggered model provides a recipe to build quantum walks on generic graphs by partitioning the vertices into \textit{cliques}. An element of the partition is called a polygon and the union of polygons is called a tessellation. Nontrivial SQWs use at least two tessellations, but depending on the graph more tessellations may be required, for instance, the hypercubic lattice.  For the sake of simplicity, we address only connected 2-tessellable graphs and we prove that a graph is 2-tessellable if and only if its \textit{clique graph} is 2-\textit{colorable}.

This work, besides reviewing some aspects of the SQW model proposed in Ref.~\cite{PSFG15}, characterizes the class of graphs on which 2-tessellable SQWs are equivalent to some Szegedy or coined QWs. We show that the Szegedy and coined QW models are defined in a restricted class of graphs, which is included in the class of \textit{line graphs} of \textit{bipartite graphs}. We prove that a SQW with no edge in the tessellation intersection can be cast into the extended Szegedy QW model. SQWs that are equivalent to Szegedy's QWs inherits the results regarding the advantage of Szegedy's QWs over their classical counterparts, see for instance Refs.~\cite{Szegedy:2004,KMOR15}. On the other hand, a SQW using two tessellations on graphs that are not \textit{line graphs} of bipartite graphs cannot be reduced to Szegedy's QW model.

One of the main applications of QWs is the spatial search problem, whose data are spread out in a lattice, for instance, and each step costs some resource. It was discussed initially by Benioff~\cite{Ben02}, who tried without success to use Grover's algorithm to beat random-walk-based searches. The first efficient quantum-walk-based search seems to be the one in Ref.~\cite{Shenvi:2003}, which used the DTQW model to search a marked vertex in a hypercube. In the coined QW model, the coin used on the non-marked vertices is the Grover coin and the coin used on the marked vertices is the minus identity operator $(-I)$. In the SQW model, vertices are marked using partial tessellations. Ref.~\cite{Ambainis:2013} used a SQW with three tessellations to search for a single marked vertex in a two-dimensional criss-cross lattice (see Ref.~\cite{PSFG15}). In this work, we provide the first example of a 2-tessellable SQW search, which is shown numerically to be more efficient than random-walk-based search.

The structure of the paper is as follows. In Sec.~\ref{sec:Staggered}, we describe how to define the evolution operator of the SQW model. In Sec.~\ref{sec3}, we show that a graph is 2-tessellable if and only if its \textit{clique graph} is 2-\textit{colorable} and we discuss the classes of graphs on which the SQW model reduces to the Szegedy or coined models. In Sec.~\ref{sec4}, we characterize which SQWs can be cast into Szegedy's QW model. In Sec.~\ref{sec5}, we describe how to convert a SQW into an equivalent form in the coined QW model for graphs on a restricted class and give nontrivial examples, namely the honeycomb lattice and the three-state QW on the line. In Sec.~\ref{sec6}, we provide an example of an efficient SQW search, which cannot be put into other QW models.  In Sec.~\ref{sec7}, we draw our conclusions.  Appendix~\ref{appendixA} is a glossary of some terms in graph theory and appendix~\ref{appendixB} gives a formal definition of Szegedy's QW.

\section{Defining the evolution operator}\label{sec:Staggered}

A quantum walk model is a recipe for building an evolution operator based on local unitary operators. Local operators obey the graph structure in the sense that if a particle is on vertex $v$, it can move only to the vertices that are in the neighborhood of $v$. In the discrete-time models, one step of the quantum walk is a product of such operators. The flip-flop coined QW model has an internal space, which can become explicit when the coin is an orthogonal reflection by converting the coin directions into extra vertices as described in Ref.~\cite{Por16}. In this case, the Hilbert space is spanned by the vertices of the extended graph. In this work, we address only QWs on Hilbert space that are spanned by the vertices of the graph. The following models are included in this analysis: (1)~the flip-flop coined QW models with coins that are orthogonal reflections, (2)~Szegedy's QW model, and (3)~the staggered QW model.

Let us start with an example of a SQW that is included neither in Szegedy's model nor in the coined model. The recipe to build the SQW on the graph of Fig.~\ref{graph1.pdf} is as follows.

\begin{figure}[h!] 
\centering
\includegraphics[scale=0.60]{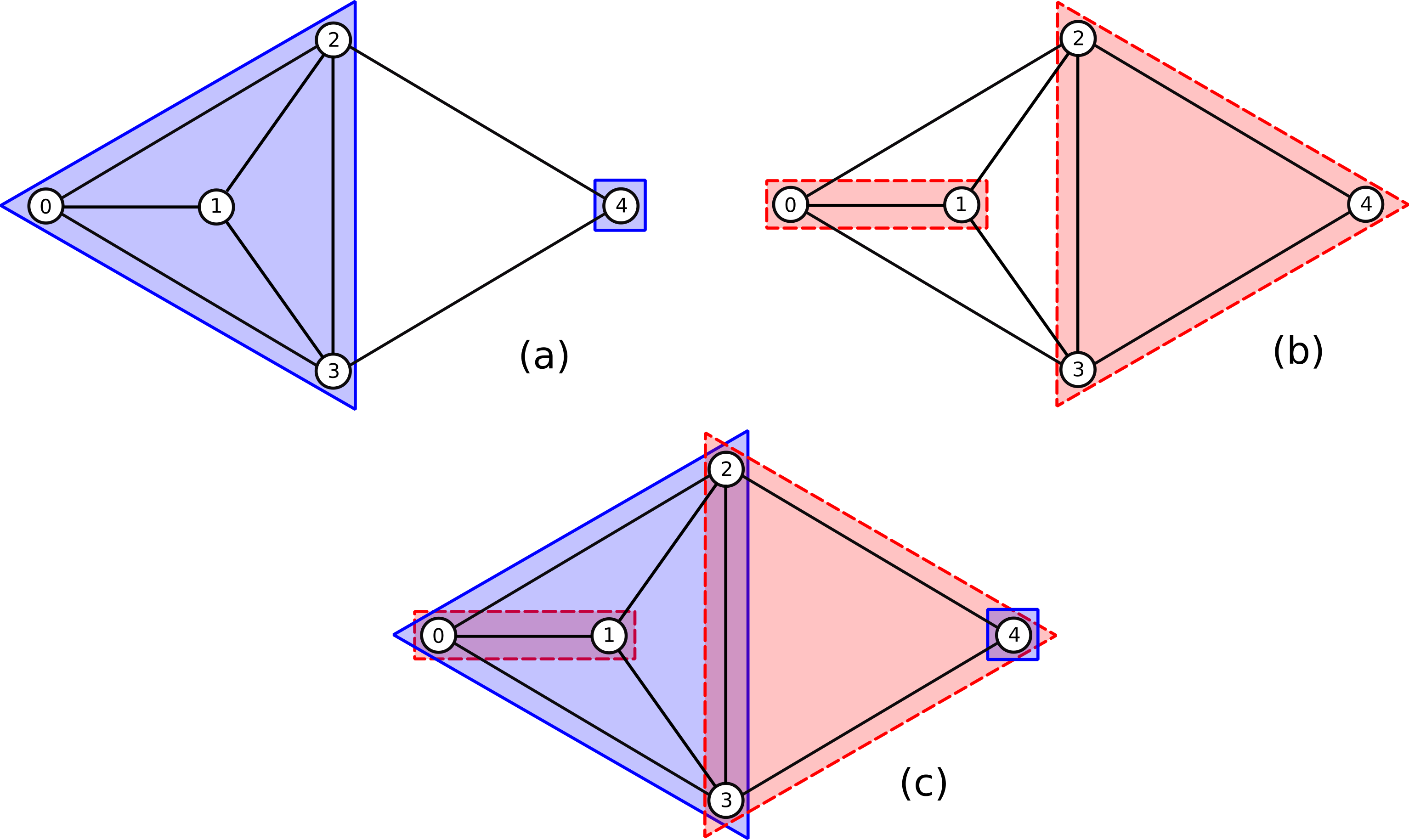}
\caption{Procedure to define a SWQ on a graph. Panel~(a) describes the blue tessellation. Panel~(b) describes the red tessellation. Panel~(c) presents the tessellation union, which must cover all edges.} 
\label{graph1.pdf}
\end{figure}

\textit{Step 1.} Make a partition of the vertices so that each element of the partition is a \textit{clique}. An element of the partition is called a polygon. The polygons do not overlap and their union contains all vertices. The polygon union is called a tessellation and we use the blue color for the first graph tessellation. Notice that some edges are not in the blue tessellation, as we can check in Fig.~\ref{graph1.pdf}(a). A polygon is always a clique, but not necessarily a \textit{maximal clique}. An edge is in a polygon if and only if the endpoints of the edge are in the polygon.

\textit{Step 2.} Associate a unit vector to each polygon so that the vector belongs to the subspace spanned by the vertices of the polygon. The simplest choice is the uniform superposition given by 
\begin{align*}
	&\ket{\alpha_0} = \frac{1}{2}\left( \ket{0}+\ket{1}+\ket{2}+\ket{3}\right), \\
	&\ket{\alpha_1} = \ket{4},
\end{align*}
where $\ket{\alpha_0}$ is associated with (or induces) the blue polygon in the form of a triangle in Fig.~\ref{graph1.pdf}(a) and $\ket{\alpha_1}$ is associated with (or induces) the square. Any other choice so that $\ket{\alpha_0}$ and $\ket{\alpha_1}$ have no zero entry and unit $l^2$-norm is acceptable.  Now we are ready to define the first local unitary operator, which has the following expression
\begin{equation}
U_0\,=\,2\,\ket{\alpha_0}\bra{\alpha_0} +2\,\ket{\alpha_1}\bra{\alpha_1}- I.
\end{equation}
By construction, $U_0$ is unitary and Hermitian ($U_0^2=I$) because $\braket{\alpha_j}{\alpha_{j'}}=\delta_{jj'}$ for $0\le j,j'\le 1$. $U_0$ is local because the particle does not leave the polygon. Since a polygon is a clique, the particle can move only to neighboring vertices because $\braket{\alpha_j}{j'}=0$ if $j'$ is not a vertex of the polygon induced by $\alpha_j$ and, if $j'$ is a vertex in $\alpha_j$, $U_0\ket{j'}$ belongs to the subspace spanned by the vertices of $\alpha_j$.

\textit{Step 3.} Now we are going to make a second vertex partition in order to cover the edges that were not included in the first tessellation. Fig.~\ref{graph1.pdf}(b) shows that this task is doable for this graph. In the generic case, we may need to use more than two tessellations.

\textit{Step 4.} Similar to Step 2, we associate a unit vector in the subspace spanned by the polygon vertices to each polygon. Again, we use the uniform superposition 
\begin{align*}
	&\ket{\beta_0} = \frac{1}{\sqrt 2}\left(\ket{0}+\ket{1}\right),\\
	&\ket{\beta_1} = \frac{1}{\sqrt 3}\left( \ket{2}+\ket{3}+\ket{4}\right). 
\end{align*}
The second local unitary operator is
\begin{equation}
U_1\,=\,2\,\ket{\beta_0}\bra{\beta_0} +2\,\ket{\beta_1}\bra{\beta_1}- I.
\end{equation}

\textit{Step 5.} Since we have covered all edges of the graph, as we can check in Fig.~\ref{graph1.pdf}(c), the evolution operator is given by
\begin{equation}
U\,=\,U_1 U_0= 
\frac{1}{6}\left[ 
\begin{array}{ccccc} 
\,\,\,\,\,3&-3&3&3&0\\ 
-3&3&3&3&0\\ 
\,\,\,\,\,1&1&3&-3&4\\ 
\,\,\,\,\,1&1&-3&3&4\\ 
\,\,\,\,\,4&4&0&0&-2
\end{array} 
\right] 
.
\end{equation}

\

Now we can generalize this construction for a generic simple undirected graph $\Gamma$, whose edges can be covered by two tessellations called $\alpha$ and $\beta$ \footnote{We use the following convention in this work: Polygons of tessellation $\alpha$ are surrounded by continuous boundaries and filled with transparent blue color. Polygons of tessellation $\beta$ are surrounded by dashed boundaries and filled with transparent red color. The color choices play no relevant role.}. The evolution operator is
\begin{equation}\label{U}
    U \,=\, U_1 U_0,
\end{equation}
where
\begin{eqnarray}
  U_0 &=& 2\sum_{k=0}^{m-1} \ket{\alpha_k}\bra{\alpha_k} - I, \label{U_0}\\
  U_1 &=& 2\sum_{k=0}^{n-1} \ket{\beta_k}\bra{\beta_k} - I, \label{U_1}
\end{eqnarray}
and $m$ and $n$ are the number of polygons in each tessellation, and
\begin{eqnarray}
  \ket{\alpha_k} &=&  \sum_{k'\in \alpha_k} a_{k\,k'} \ket{k'}, \label{alpha_k} \\
  \ket{\beta_k} &=&  \sum_{k'\in \beta_k} b_{k\,k'} \ket{k'}, \label{beta_k}
\end{eqnarray}
where $a_{k,k'}$ are nonzero complex amplitudes for $k'\in \alpha_k$, which means that if $k'$ is a vertex of the polygon induced by vector $\ket{\alpha_k}$ then $a_{k,k'}\neq 0$ otherwise  $a_{k,k'}=0$;  likewise $b_{k,k'}$  are nonzero complex amplitudes for $k'\in \beta_k$ and zero otherwise. Index $k'$ in $a_{kk'}$ and $b_{kk'}$ runs from 0 to $N-1$, where $N$ is the number of vertices of $\Gamma$.

Formally, a tessellation is a partition of the graph into \textit{cliques}, that is, each element of the partition is a clique and two elements of the partition cannot have a vertex in common. An element of the partition is called a polygon. We can associate unitary and Hermitian operators with a tessellation of $m$ polygons using the form of the operator given by Eqs.~(\ref{U_0}) and~(\ref{alpha_k}). A unitary and Hermitian operator is called an \textit{orthogonal reflection} if it is associated with (or induces) a tessellation. A SQW is defined by an evolution operator that is a product of orthogonal reflections such that the union of the tessellations induced by the orthogonal reflections covers the edges of the graph~\cite{PSFG15,Por16}.

An interesting class of orthogonal reflections is obtained using polygons in uniform superposition. In this case, Eq.~(\ref{alpha_k}) reduces to
\begin{equation}
\ket{\alpha_k} \,=\, \frac{1}{\sqrt{|\alpha_k|}} \sum_{k'\in \alpha_k} \ket{k'},
\end{equation}
where $|\alpha_k|$ is the number of vertices in polygon $\alpha_k$. This class generalizes flip-flop DTQWs with the Grover coin in the sense that, when we convert the internal degrees of freedom into extra vertices, the extended graph on which the Grover walk is defined is a \textit{line graph} of a \textit{bipartite graph} while the SQW can be defined on a wider class.

\section{Main classes of graphs}\label{sec3}

In this section, we use graph theory to classify the most relevant classes of graphs that help to identify SQWs that can be reduced into the Szegedy or coined model.  Since our focus in this work is the set of 2-tessellable SQWs (those that have an evolution operator that is the product of exactly two orthogonal reflections), we start addressing the following question: Which graphs are 2-tessellable? We give a necessary condition. Each vertex of the graph must belong to at most two \textit{maximal cliques}. For example, graph $\Gamma_9$ in Fig.~\ref{fig:app4} has a central vertex that belongs to the intersection of five \textit{maximal cliques}. To define a SQW on this graph we need to employ at least five tessellations. By inspection, we can check that only graphs $\Gamma_2$, $\Gamma_3$, $\Gamma_4$, $\Gamma_5$, and $\Gamma_6$ are 2-tessellable and those graphs obey the necessary condition.  This condition is not sufficient. For instance, each vertex of $\Gamma_7$ belongs to two \textit{maximal cliques}, but we need to employ at least three tessellations to define a SQW on this graph.

\begin{figure}[h!] 
\centering
\includegraphics[scale=0.45]{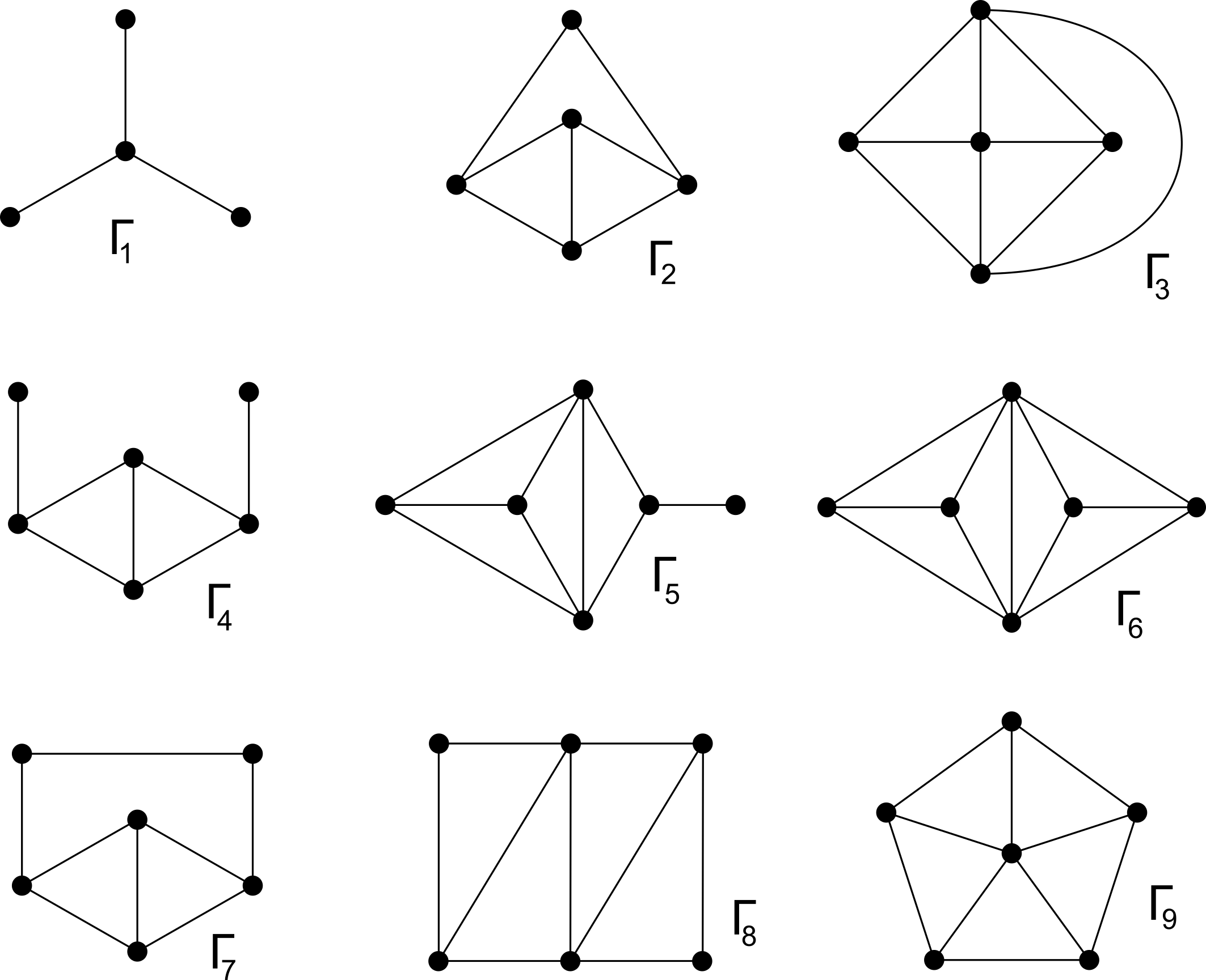}
\caption{Nine forbidden Beineke \textit{induced subgraphs}.}
\label{fig:app4}
\end{figure}

To definitively answer the question, we need to use the \textit{clique graph} of the original graph because the \textit{clique graph} contains all information about the adjacency relation of the \textit{maximal cliques} of the original graph. The necessary and sufficient condition is the \textit{clique graph} must be 2-\textit{colorable}. To prove this we need the following lemma:
\begin{lemma}\label{lemma1}
Each \textit{maximal clique} of a $2$-tessellable graph is inside a polygon (blue or red). 
\end{lemma} 
\begin{proof}
Suppose that the graph is a clique. We state that one tessellation must have a polygon that covers the whole clique. In fact, if a blue polygon does not cover the whole clique, the missed vertices are \textit{adjacent} to all vertices of that blue polygon. The red polygon must cover all edges of the whole clique because, otherwise, there will be at least one edge whose endpoints belong to polygons with different colors. This edge will not be in the tessellation union violating the definition of 2-tessellable SQWs. Then, all vertices of the whole clique must be in the red polygon. The same argument is valid for each \textit{maximal clique} of a generic 2-tessellable graph.  \qed
\end{proof}

\begin{proposition}\label{prop:2tess}
A connected graph is $2$-tessellable if and only if its \textit{clique graph} is $2$-\textit{colorable}.
\end{proposition} 
\begin{proof}
Let us start with the sufficiency. If the \textit{clique graph} is 1-colorable, the original graph is a clique, which is 2-tessellable. If the \textit{clique graph} is 2-colorable, we can use the coloration induced by the \textit{clique graph} on the original graph, which generates a partial tessellation of the original graph with two colors (some of the vertices belong to two different-colored polygons), and then we complete the partial tessellations by defining new polygons that will cover the remaining vertices (those that do not belong to the intersection of the induced colorations). In the end, all vertices belong to two different-colored polygons. This proves the sufficiency.
 
To prove the necessity, we take the original graph with two tessellations and erase all polygons that are not maximal cliques. The union of the remaining polygons still covers all vertices of the original graph because each vertex belongs to a maximal clique (using the above lemma here). The union of the remaining polygons also covers all edges and is a 2-colorable \textit{clique cover}. Those remaining polygons induce a 2-coloration of the \textit{clique graph} unless the whole graph is a clique whose \textit{clique graph} is 1-colorable. \qed
\end{proof}

This proposition shows that graph $\Gamma_7$ in Fig.~\ref{fig:app4} is not 2-tessellable because the \textit{clique graph} of $\Gamma_7$ is a pentagon, which is not 2-colorable. 

It seems that Proposition~\ref{prop:2tess} cannot be easily extended for graphs that require more than two tessellations. Take for example the Haj\'os graph (a triangle surrounded by three triangles each one sharing an edge in common with the central triangle). This graph is 3-tessellable and its clique graph is not 3-colorable~\cite{Szw03}. Besides, it is possible to have a maximal clique that does not belong to any polygon. 


\subsection*{Class 1\,\, Graphs that are not line graphs}

The nine forbidden Beineke \textit{induced subgraphs} of Fig.~\ref{fig:app4} are used to check whether a graph is a \textit{line graph}. Beineke~\cite{Bei70} has proved the following theorem:
\begin{theorem}
Let $\Gamma'$ be a graph. There exists a graph $\Gamma$ such that $\Gamma'$ is the line graph of $\Gamma$ if and only if $\Gamma'$ contains no graph of Fig.~\ref{fig:app4} as an \textit{induced subgraph}.
\end{theorem} 

If we want to define a 2-tessellable SQW on a graph that is not a line graph (Class~1), our graph on the one hand cannot have $\Gamma_1$, $\Gamma_7$, $\Gamma_8$, and $\Gamma_9$ as an \textit{induced subgraph} (their \textit{clique graphs} are not 2-colorable) and on the other hand it must have $\Gamma_2$, $\Gamma_3$, $\Gamma_4$, $\Gamma_5$, or $\Gamma_6$ as an \textit{induced subgraph} (Beineke's theorem). 

The 2-tessellable SQWs on graphs in Class~1 have the following property: There are necessarily one or more edges in the intersection of the tessellations. This follows from the fact that $\Gamma_2$, $\Gamma_4$, $\Gamma_5$, and $\Gamma_6$ have \textit{maximal cliques} sharing an edge and  $\Gamma_3$  have \textit{maximal cliques} sharing three edges. Take $\Gamma_6$ for instance, it comprises two 4-cliques having in common the central vertical edge.

SQWs on graphs in Class~1 can be included neither in Szegedy's model nor in the flip-flop coined model, because Refs.~\cite{PSFG15,Por16} have shown that any Szegedy's QW or flip-flop coined walk using orthogonal reflections are equivalent to SQWs on line graphs of the bipartite graphs. 

\subsection*{Class 2a\,\, Line graphs of nonbipartite graphs}

To characterize this class of graphs we use the Krausz partition.
\begin{definition}
A Krausz partition is a collection $C$ of subgraphs of a graph $\Gamma$ that satisfies the following items: (1)~each element of $C$ is a clique, (2) each edge of $\Gamma$ is in exactly one element of $C$, and (3)~each vertex is in exactly two elements of $C$.
\end{definition}
For example, Fig.~\ref{fig:app2} shows the Krausz partition of graph $\Gamma'$ (the same as Fig.~\ref{graph1.pdf}). Krausz~\cite{Kra43} has proved the following theorem:
\begin{theorem}
A graph $\Gamma'$ is a line graph of some graph $\Gamma$ if and only if $\Gamma'$ has a Krausz partition.
\end{theorem}
Krausz's theorem is an alternative way to check whether a graph is a line graph. The advantage of using the Krausz partition is that it shows how to obtain the \textit{root graph}. Each element of the Krausz partition is associated with a vertex of the root graph. Two vertices of the root graph are adjacent if and only if the intersection of the corresponding elements of the partition is nonempty. On the other hand, it is known that a generic graph is \textit{bipartite} if and only it is 2-colorable. Then, the line graph of a graph $\Gamma$ has a 2-colorable Krausz partition if and only if $\Gamma$ is \textit{bipartite}.

\begin{figure}[h!] 
\centering
\includegraphics[scale=0.30]{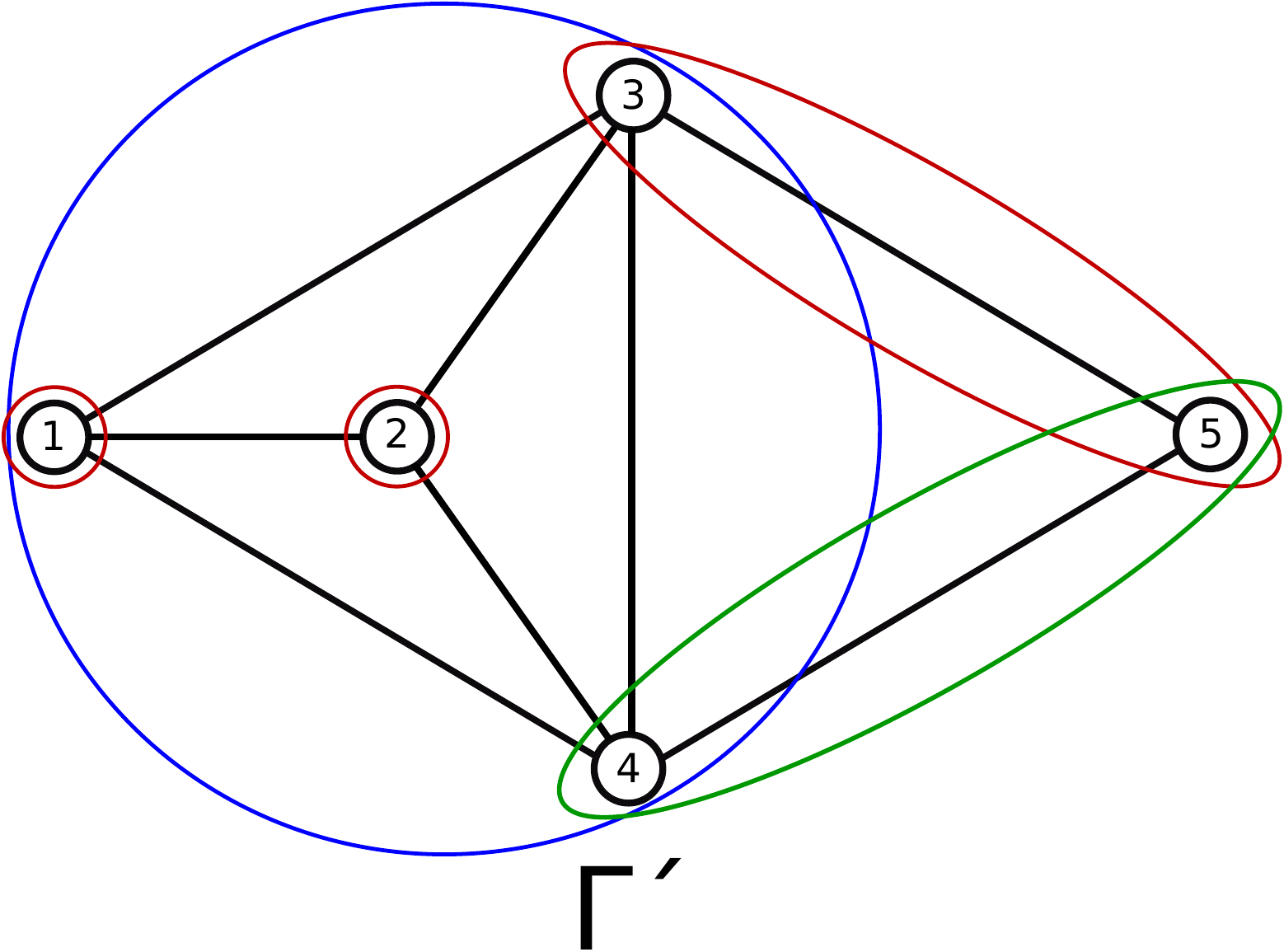}\hspace{0.3cm}
\includegraphics[scale=0.30]{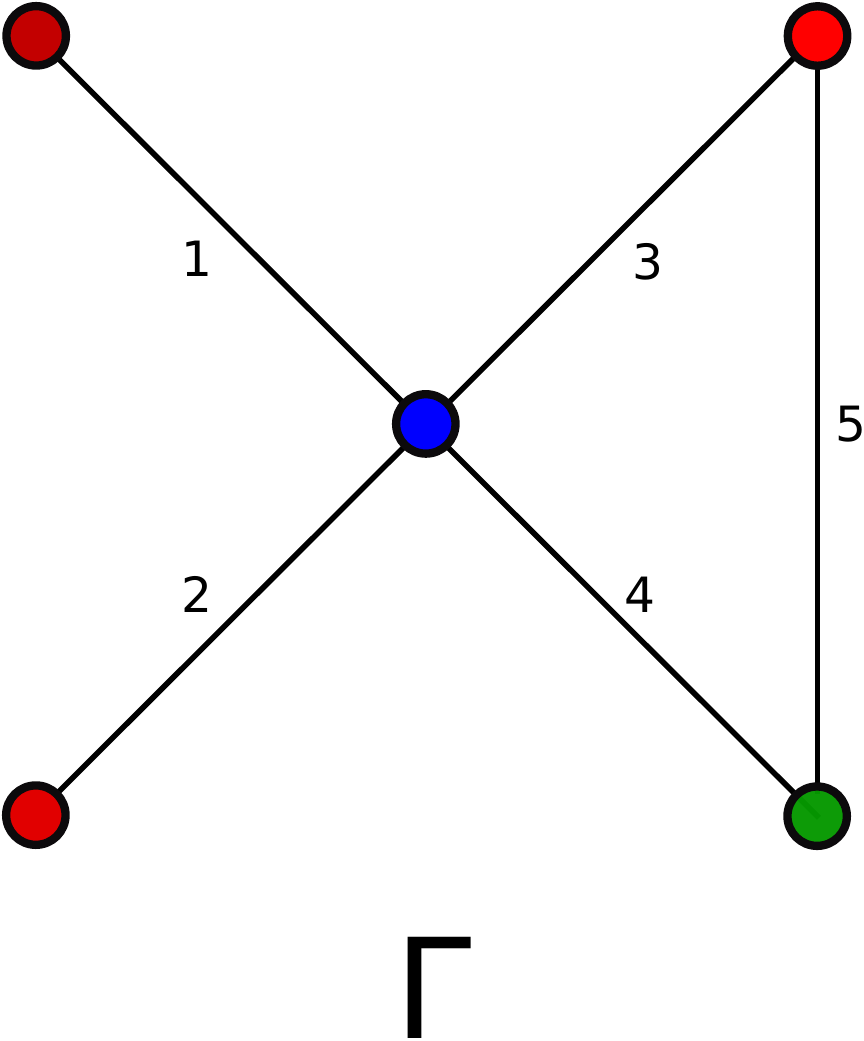}
\caption{A \textit{line graph} $\Gamma'$ with the Krausz partition and its \textit{root graph} $\Gamma$. There is an one-to-one map from the Krausz partition and the vertices of the root graph displayed by the colors. Notice that the Krausz partition is not 2-\textit{colorable} and $\Gamma$ is not \textit{bipartite}.}
\label{fig:app2}
\end{figure}

We state that the class of graphs that are line graphs of nonbipartite graphs is characterized by graphs that have Krausz partition that are not 2-colorable. If a graph of this class has a 2-colorable \textit{clique graph}, then the graph is 2-tessellable. The tessellation union cannot coincide with the Krausz partition because the Krausz partition is not 2-colorable. The intersection of two elements of the Krausz partition has no edge. If the graph is 2-tessellable, there must be an edge in a polygon intersection because the tessellation union splits the graph into \textit{maximal cliques} so that each vertex is in exactly two cliques. The only escape using the definition of Krausz partition is the following: There is an edge in the intersection of two polygons. This property is shared with graphs that are not line graphs (Class~1).

Besides the correspondence between the elements of the Krausz partition and the vertices of the root graph, there is another correspondence, coming from the definition of line graphs, between the edges of the root graph and the vertices of the line graph. For example, there is an one-to-one map between $E(\Gamma)$ and $V(\Gamma')$ of graphs $\Gamma$ and $\Gamma'$ of Fig.~\ref{fig:app2},  which is displayed by the numerical labels. 

Summing up, we have proved so far that 2-tessellable SQWs on graphs in Class~1 or~2a have at least one edge in the tessellation intersection. Those SQWs cannot be cast into Szegedy's QW model because they are neither line graphs nor line graphs of bipartite graphs. Ref.~\cite{PSFG15} has shown that any Szegedy's QW is equivalent to a SQW on the line graph of a bipartite graph.

There are graphs in Class~2a that are not 2-tessellable, for instance, the Haj\'os graph.

\subsection*{Class 2b\,\, Line graphs of bipartite graphs} 

Class~2b is comprised by graphs that have 2-colorable Krausz partition. Another way to characterize graphs in this class is by using the following theorem (see for instance Ref.~\cite{Pet03}):
\begin{theorem}
$\Gamma$ is the line graph of a bipartite graph if and only if $\Gamma$ is diamond-free and the clique graph $K(\Gamma)$ is bipartite.
\end{theorem} 
Ref.~\cite{Pet03} also has the following result:
\begin{proposition}
A graph is diamond-free if and only if any two maximal cliques intersect
in at most one vertex, which holds if and only if each edge lies in exactly one
maximal clique.
\end{proposition}


Graph $\Gamma$ in Class~2b cannot have a \textit{diamond} as an induced subgraph and two maximal cliques of $\Gamma$ do not share a common edge. This means that any \textit{minimum clique cover} of $\Gamma$ is also a \textit{minimum clique partition}. All graphs in this class are 2-tessellable because $K(\Gamma)$ is bipartite and the 2-colorable Krausz partition induces the blue and red tessellations. A SQW that use the tessellations induced by the Krausz partition can be cast into Szegedy's model, as will be shown in Sec.~\ref{sec4}.

To define a SQW on $\Gamma$ that is not included in Szegedy's model, we have to choose tessellations blue and red with an edge in the tessellation intersection. For example, the tessellations of the graph in Fig.~\ref{fig:linegraphofbipart}(a) have two edges in the tessellation intersection while the tessellations of the same graph in Fig.~\ref{fig:linegraphofbipart}(b) have no edges in the tessellation intersection. The tessellation union of Fig.~\ref{fig:linegraphofbipart}(b) is a 2-colorable Krausz partition confirming that this graph is in Class~2b.

\begin{figure}[ht!] 
\centering
\includegraphics[scale=0.45]{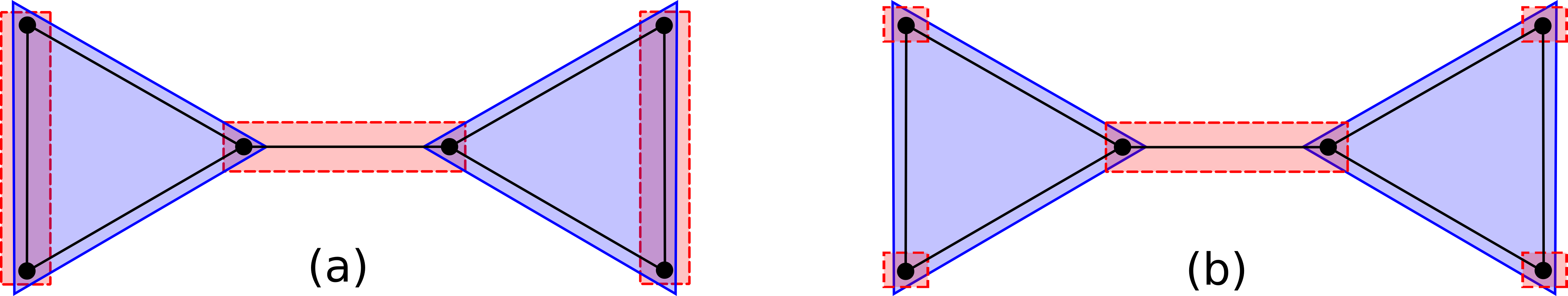}
\caption{A graph (a barbell) with two different tessellation patterns. A SQW with tessellation~(a) does not belong to Szegedy's model while a SQW with tessellation~(b) always belongs to Szegedy's model.}
\label{fig:linegraphofbipart}
\end{figure}

If we remove the elements of the Krausz partition that are not maximal cliques, we have a 2-colorable \textit{minimum clique partition}. We must use this partition to define the blue and red polygons of a 2-tessellable SWQ because each maximal clique must be inside a polygon (Lemma~\ref{lemma1}). After this procedure, if there are vertices that do not belong to two polygons, then we have to define new polygons until both tessellations are complete (all vertices must be in exactly two polygons). If there are two or more vertices in the same blue polygon of the 2-colorable \textit{minimum clique partition} that do not belong to red polygons, we can put those vertices in a single red polygon to have one or more edges in the tessellation intersection. Such SQWs cannot be cast into Szegedy's model. If there is at most one vertex in each polygon of the 2-colorable \textit{minimum clique partition} that does not belong to a second polygon (with a different color), then the SQW in this case can be cast into Szegedy's model, as will be shown in Sec.~\ref{sec4}.

\subsection*{Class~2b$'$\,\, Line graphs of bipartite graphs with perfect matching}

Let us define subclass Class~2b$'$ (we will show that Class~2b$'$~$\subset$ Class~2b) of graphs $\Gamma$ that obey the following three conditions: 
\begin{itemize}
\item[(1)]~$\Gamma$ has a \textit{perfect matching} $M$. \label{3conditions}
\item[(2)]~The endpoints of an edge in $M$ induces a \textit{maximal clique} of size two. 
\item[(3)]~If we delete the edges of the perfect matching (leaving the endpoints), we obtain a union of \textit{disconnected} maximal cliques.
\end{itemize} 
\noindent
The third condition is equivalent to state that the complement of the perfect matching $M$ in $\Gamma$ is a union of disconnected maximal cliques, some (or all) of them can have one vertex.

Let us show that a graph $\Gamma$ in Class~2b$'$ is the (or isomorphic to the) line graph of a bipartite graph. Conditions~(1) to~(3) guarantee that $\Gamma$ has a 2-colorable Krausz partition. The edges of the perfect matching are the elements of the Krausz partition with the red color and the maximal cliques in the complement of the perfect matching in $\Gamma$ are the elements with the blue color. Each vertex is in exactly two elements of the partition (elements with different colors) establishing a well-defined Krausz partition. This Krausz partition is 2-colorable, then $\Gamma$ is the line graph of a bipartite graph. This proves that Class~2b$'$~$\subset$ Class~2b.

If we add an extra condition, we can prove a stronger result: A graph $\Gamma$ that obeys conditions~(1) to~(3) and has no vertex with \textit{degree} 1 is the line graph of its clique graph, that is, $\Gamma=L(K(\Gamma))$. Besides, $K(\Gamma)$ is bipartite. Let us show that $\Gamma$ is the line graph of $K(\Gamma)$. This follows from the fact that the vertices of $K(\Gamma)$ correspond to all elements of the Krausz partition of $\Gamma$ because there is no element with one vertex. Each element of the Krausz partition has at least two vertices. Besides, each element is a maximal clique. By the definition of clique graph, two vertices of $K(\Gamma)$ are adjacent if and only if the intersection of the corresponding elements of  the Krausz partition is nonempty. This shows that $K(\Gamma)$ is the root graph of $\Gamma$.


\section{Staggered QWs that are in Szegedy's model}\label{sec4}

Graphs in Class~2b have 2-colorable Krausz partitions. If we use the Krausz partition to define the tessellations (there is exactly one vertex in the polygon intersections), then the SQWs are included in the extended Szegedy QW model (see Appendix~\ref{appendixB} for formal definition). The following proposition generalizes this statement. 
\begin{proposition}\label{mainproposition}
A $2$-tessellable SQW with no edge in the intersection of the tessellations can be cast into the extended Szegedy QW model.
\end{proposition} 
\begin{proof} 
Suppose that the 2-tessellable SQW is defined on a graph $\Gamma'$ with $N$ vertices. The labels of the vertex set of $\Gamma'$ run from 0 to $N-1$ and the basis of the Hilbert space associated with $\Gamma'$ is $\big\{\ket{0},...,\ket{N-1}\big\}$.

We have shown in Sec.~\ref{sec3} that $\Gamma'$ belongs neither to Class~1 nor Class~2a because any 2-tessellable SQW on graphs in those classes has an edge in the tessellation intersection. Then, $\Gamma'$ must be in Class~2b, that is,  $\Gamma'$ is the line graph of a bipartite graph $\Gamma$.


Since $\Gamma'$ is 2-tessellable, each maximal clique is inside a blue or red polygon (Lemma~\ref{lemma1}). Adjacent maximal cliques have different color. Since there is no edge in the tessellation intersection (there is exactly one vertex in the polygon intersections), the tessellation union is a 2-colorable Krausz partition. The root graph $\Gamma$ is bipartite with edges connecting set $A$ (associated with blue polygons) and set $B$ (associated with red polygons). No two vertices in $A$ are adjacent and the same for $B$. Let $m,n$ be the number of polygons in the blue and red tessellations, respectively. We label the vertices in $A$ by $\tilde\alpha_k$, $0\le k<m$, and the vertices in $B$ by $\tilde\beta_k$, $0\le k<n$ using the one-to-one mapping between the elements of the Krausz partition and the vertices of $\Gamma$. Each vertex $k$ of $\Gamma'$ belongs to the intersection of two polygons
$\alpha_i$ and $\beta_j$ for some $0\le i<m$ and $0\le j<n$. Then, there is an one-to-one mapping between the vertex set of $\Gamma'$ and the edge set of $\Gamma$ given by
\begin{equation}\label{bijection}
{k}\leftrightarrow(\tilde\alpha_i,\tilde\beta_j)
\end{equation}
where $k$ is a vertex of $\Gamma'$ and $\tilde\alpha_i,\tilde\beta_j$ are the endpoints of the edge of $\Gamma$ that corresponds to $k$. Label $k$ runs from 0 to $N-1$ and labels $i,j$ run over the edge set $E(\Gamma)$, whose cardinality is equal to $N$. Labels $i,j$ run in the whole range $0\le i<m$ and $0\le j<n$ if and only if $\Gamma$ is the \textit{complete bipartite graph}.

Let $\ket{\tilde\alpha_i}\otimes\ket{\tilde\beta_j}\equiv\ket{\tilde\alpha_i,\tilde\beta_j}$, for $0\le i<m$ and $0\le j<n$, be the computational basis of ${\cal H}^m\otimes{\cal H}^n$ using weird labels, which help to remember the correspondence between the vertices $\tilde\alpha_i$ and $\tilde\beta_j$ of $\Gamma$ and polygons $\alpha_i$ and $\beta_j$ of $\Gamma'$.

Using bijection~(\ref{bijection}), define the linear transformation $T:{\cal H}^N\rightarrow{\cal H}^m\otimes{\cal H}^n$ by 
\begin{equation}\label{T}
T\ket{k}=\ket{\tilde\alpha_i}\otimes\ket{\tilde\beta_j}.
\end{equation}
Using $T$, define the unitary operators $R_0$ and $R_1$ in ${\cal H}^m\otimes{\cal H}^n$ by
\begin{eqnarray}
  R_0 &=& 2\sum_{k=0}^{m-1} \ket{\phi_{k}}\bra{\phi_{k}} - I, \label{R_0}\\
  R_1 &=& 2\sum_{k=0}^{n-1} \ket{\psi_{k}}\bra{\psi_{k}} - I, \label{R_1}
\end{eqnarray}
where 
\begin{eqnarray}
  \ket{\phi_{k}} &=& T\ket{\alpha_k} \,=\, \sum_{k'\in\alpha_k} a_{k\,k'} \,T\ket{k'}, \label{pphi_k} \\
  \ket{\psi_{k}} &=& T\ket{\beta_k} \,=\, \sum_{k'\in\beta_k} b_{k\,k'} \,T\ket{k'}, \label{ppsi_k}
\end{eqnarray}
$a_{kk'}$ and $b_{kk}$ are the entries of $\ket{\alpha_k}$ and $\ket{\beta_k}$ given by Eqs.~(\ref{alpha_k}) and~(\ref{beta_k}). Now, let us show that $W=R_1 R_0$ is the evolution operator of a well-defined extended Szegedy QW on $\Gamma$ equivalent to the SQW on $\Gamma'$. The Hilbert space associated with $\Gamma$ is ${\cal H}^m\otimes{\cal H}^n$.

Expressing $(\ref{pphi_k})$ and $(\ref{ppsi_k})$ in the computational basis of ${\cal H}^m\otimes{\cal H}^n$, we obtain
\begin{eqnarray}
  \ket{\phi_{k}} &=& \,\,\sum_{\mathclap{\substack{j\\ \textrm{ such that}\\(\tilde\alpha_k,\tilde\beta_j)\in E(\Gamma)}}} \,\,\,\, a_{k;(k,j)} \,\ket{\tilde\alpha_k,\tilde\beta_j}, \label{phi_k} \\
  \ket{\psi_{k}} &=&  \,\,\sum_{\mathclap{\substack{i\\ \textrm{ such that}\\(\tilde\alpha_i,\tilde\beta_k)\in E(\Gamma)}}} \,\,\,\, {b}_{k;(i,k)} \,\ket{\tilde\alpha_i,\tilde\beta_k}, \label{psi_k}
\end{eqnarray}
where index $k$ of $\ket{\phi_{k}}$ runs from 0 to $m-1$ and index $k$ of $\ket{\psi_{k}}$ runs from 0 to $n-1$. The notation $a_{k;(i,j)}$ means that $a_{k;(i,j)}=a_{k;k'}$ for the value of $k'$ such that $T\ket{k'}= \ket{\tilde\alpha_i,\tilde\beta_j}$. Since $k'$ in Eq.~(\ref{pphi_k}) is in polygon $\alpha_k$, $T\ket{k'}=\ket{\tilde\alpha_k,\tilde\beta_j}$ for some $0\le j<n$. Since $k'$ in Eq.~(\ref{ppsi_k}) is in polygon $\beta_k$, $T\ket{k'}=\ket{\tilde\alpha_i,\tilde\beta_k}$ for some $0\le i<m$. The same notation applies to $b_{k;(i,j)}$.

Using $\ket{\phi_k}$ and $\ket{\psi_k}$, mapping (\ref{bijection}), and the notation of Appendix~\ref{appendixB}, we define matrices $P$ and $Q$ whose dimensions are $m\times n$ and $n\times m$ and whose entries are $p_{kj}=\left|a_{k;(k,j)}\right|^2$ and $q_{ki}=\left|b_{k;(i,k)}\right|^2$, respectively. $P$ and $Q$ are right-stochastic matrices. In fact, $\sum_{j=0}^{n-1} p_{kj}=1$, for $0\le k<m$ and $\sum_{i=0}^{m-1} q_{ki}=1$, for $0\le k<n$ because $\ket{\alpha_k}$ and $\ket{\beta_k}$ have unit $l_2$-norm.  Let $P',Q'$ be the matrices obtained from $P,Q$ by replacing the nonzero entries with 1. We have to show that $P'=Q'^{\textrm{T}}$ (see Definition \ref{def:SzegedyQW} and matrix~(\ref{biadmatrix})), or equivalently $p'_{kj}=1$ $\Leftrightarrow$ $q'_{jk}=1$. Suppose that $p'_{kj}=1$. Then $a_{k;(k,j)}\neq 0$, where $a_{k;(k,j)}$ is the coefficient of some $\ket{k'}$ in $\ket{\alpha_k}$. Vertex $k'$ in $\Gamma'$ is in the intersection of polygons $\alpha_k$ and $\beta_j$. The coefficient of $\ket{k'}$ in $\beta_j$ also must be nonzero because $k'\in\beta_j$. Then $b_{j;(k,j)}\neq 0$ and $q'_{jk}=1$. The same argument works the other way around and if $q'_{jk}=1$ then $p'_{kj}=1$.

Now let us display the connection between the evolution operator $U$ of the staggered model given by Eq.~(\ref{U}) and $W=R_1R_0$ of Szegedy's model. $U$ and $W$ are not exactly equal in general because $W$ may have an idle subspace. Suppose that the first vectors of the computational basis of ${\cal H}^m\otimes{\cal H}^n$ are the vectors $\ket{\tilde\alpha_i,\tilde\beta_j}$ in the same order of the elements of the computational basis of ${\cal H}^N$ after using bijection~(\ref{bijection}). The order of the remaining $(m\cdot n-N)$ vectors of the computational basis of ${\cal H}^m\otimes{\cal H}^n$ does not matter.

If $\ket{\tilde\alpha_i,\tilde\beta_j}$ and $\ket{\tilde\alpha_{i'},\tilde\beta_{j'}}$ are vectors in the computational basis of ${\cal H}^m\otimes{\cal H}^n$ that correspond to $\ket{k}$ and $\ket{k'}$ in ${\cal H}^N$, then using Eqs.~(\ref{T}), (\ref{R_0}), and~(\ref{phi_k}) we obtain
\begin{equation}
\bra{\tilde\alpha_i,\tilde\beta_j}R_0\ket{\tilde\alpha_{i'},\tilde\beta_{j'}}=\bra{k}U_0\ket{k'}.
\end{equation}
This shows that the submatrix of $R_0$ obtained by selecting the first $N$ lines and columns  of $R_0$ is equal to $U_0$. If $\ket{a,b}$ is in the computational basis of ${\cal H}^m\otimes{\cal H}^n$ and $(a,b)\not\in E(\Gamma)$ then $\bra{a,b}R_0\ket{a,b}=-1$ because $\braket{a,b}{\phi_k}=0$, $\forall k$. The remaining entries of $R_0$ are zero, as we can systematically check. Summing up, we have shown that
\begin{equation}
R_0 \,=\, \left[ \begin{array}{cccc}
U_0 & 0 & \cdots & 0 \\
0 & -1 & \cdots & 0 \\
\vdots & \vdots  & \ddots & \vdots  \\
0 & 0 &\cdots & -1 \end{array} \right].
\end{equation}

It is straightforward to show that the same property holds when we compare $R_1$ and $U_1$, that is, the submatrix of $R_1$ obtained by selecting the first $N$ lines and columns  of $R_1$ is equal to $U_1$, the remaining diagonal entries of $R_1$ are $(-1)$, and the remaining nondiagonal entries are zero. Summing up, we have shown that
\begin{equation}
W \,=\, \left[ \begin{array}{cccc}
U & 0 & \cdots & 0 \\
0 & 1 & \cdots & 0 \\
\vdots & \vdots  & \ddots & \vdots  \\
0 & 0 &\cdots & 1 \end{array} \right],
\end{equation} 
where the number of 1's in the diagonal is $(m\,n-N)$, which is the dimension of the idle space in Szegedy's model. \qed

\end{proof}

The evolution operators of the Szegedy and SQW models are not exactly the same when the bipartite graph is not complete because the dimension $mn$ of the Hilbert space of Szegedy's model is larger than the dimension $N$ of the staggered model. We have shown that $U\ket{\psi}\leftrightarrow W\ket{\psi'}$ if $\ket{\psi}\leftrightarrow \ket{\psi'}$, where $\ket{\psi}\leftrightarrow \ket{\psi'}$ means that $\ket{\psi'}$ is obtained from $\ket{\psi}$ after using the bijection~(\ref{bijection}), or equivalently, $TU\ket{\psi}=WT\ket{\psi}$, $\forall \ket{\psi}\in{\cal H}^N$, where $T$ is given by Eq.~(\ref{T}).

\section{Staggered QWs that are in the coined model}\label{sec5}

The 2-tessellable SQWs on graphs $\Gamma(V,E)$ in Class~2b$'$ are included in the coined model if the tessellation that covers the perfect matching uses vectors in uniform superposition. This result follows from the demonstration of Theorem.~4.2 of Ref.~\cite{Por16}, which states that Szegedy's QWs on bipartite graphs $\Gamma(X,Y,E)$ are equivalent to a flip-flop coined QW on some $|X|$-\textit{multigraph}, if the vertices in $Y$ have degree 2 and the edges incident on the vertices in $Y$ have equal weight.

We briefly review this result using the blue and red tessellations induced by the Krausz partition of $\Gamma(V,E)$. The red tessellation covers the perfect matching. If we choose the vertex labels so that vertices in the blue polygons are consecutive numbers, the unitary operator $U_0$ associated with the blue tessellation will be a block diagonal matrix, each block associated with a blue polygon. $U_0$ is the coin operator of a QW on a new graph $\Gamma'(V',E')$ obtained from $\Gamma(V,E)$ by shrinking the blue polygons into single vertices so that a blue polygon that is a $d$-clique becomes a degree-$d$ vertex in $V'$. The cardinality of $V'$ is the number of blue polygons. The edge set $E'$ corresponds to the \textit{perfect matching} of $\Gamma$. This shrinking process can produce a \textit{multigraph}~\cite{Por16}.

The unitary operator $U_1$ associated with the red tessellation is
\begin{equation}
	U_1\,=\, 2\sum_{(i,j)\in M} \ket{\beta_{ij}}\bra{\beta_{ij}} - I,
\end{equation}
 where
\begin{equation}
	\ket{\beta_{ij}}\,=\, \frac{\ket{i}+\ket{j}}{\sqrt 2}
\end{equation}
and $M$ is the perfect matching. Simplifying $U_1$ we obtain
\begin{equation}\label{simpU_1}
	U_1\,=\, \sum_{(i,j)\in M} \ket{i}\bra{j} + \ket{j}\bra{i},
\end{equation}
which is a flip-flop shift operator on $\Gamma'(V',E')$. $U=U_1 U_0$ is the evolution operator of a well-defined flip-flop coined model on $\Gamma'(V',E')$. 

In the next subsections, we show nontrivial examples that display the connection between the staggered and coined models.

\subsection{Honeycomb lattice}

Consider the SQW on the graph $\Gamma$ of Fig.~\ref{fig:honeycomb1}, which is obtained from the hexagonal lattice (honeycomb) after replacing the vertices by triangles. Graph $\Gamma$ belongs to Class~2b$'$ because it obeys conditions~(1) to~(3). The vertex labels are chosen following a method similar to the one used for two-dimensional lattices. Recall that in the latter case, the vertex label $(x,y)$ means that it is represented by vector $x\vec e_x+y\vec e_y$, where $\vec e_x$ and $\vec e_y$ are the unit canonical vectors along axes $x$ and $y$, respectively. For the honeycomb, the unit canonical vectors must be replaced by the vectors $\vec e_x$ and $\vec e_y$ displayed in Fig.~\ref{fig:honeycomb2}. The position of half nonadjacent vertices are obtained using vectors $x\vec e_x+y\vec e_y$, for $0\le x,y < m$, where $m$ is the even number of hexagons in the $x$- or $y$-directions (we are using the cyclic or torus-like boundary conditions). The other vertices are obtained using vectors  $x\vec e_x+y\vec e_y+\vec a$, where $\vec a=(\vec e_x+\vec e_y)/3$, as shown in Fig.~\ref{fig:honeycomb2}. So, we use labels $(x,y,0)$ for the first set of vertices and $(x,y,1)$ for the second set of vertices of the honeycomb. The labels of graph $\Gamma$ of Fig.~\ref{fig:honeycomb1} requires a fourth index $k$ describing the position of the vertices inside the triangle such as $(x,y,i,k)$, for $0\le i\le 1$ and $0\le k\le 2$.

\begin{figure}[h!] 
\centering
\includegraphics[scale=0.42]{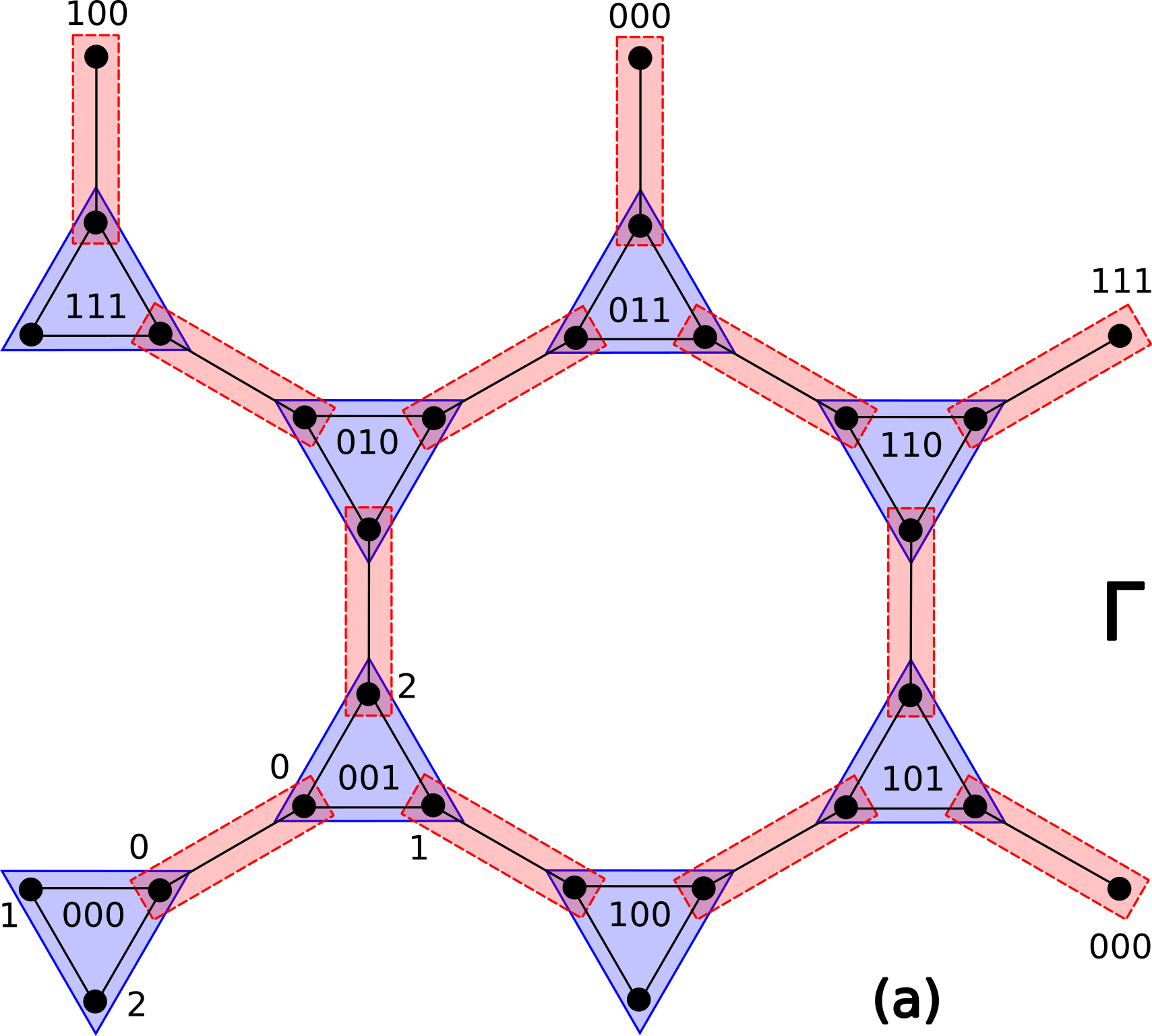}
\caption{Graph $\Gamma$ on which a SQW is defined. Labels $(x,y,i)$ used in Eq.~(\ref{alpha_xyi}) are shown inside the blue polygons. Label $k$, which runs from~0 to~2 used in Eq.~(\ref{beta_xyk}), are shown outside the blue polygons. }
\label{fig:honeycomb1}
\end{figure}

\begin{figure}[h!] 
\centering
\includegraphics[scale=0.38]{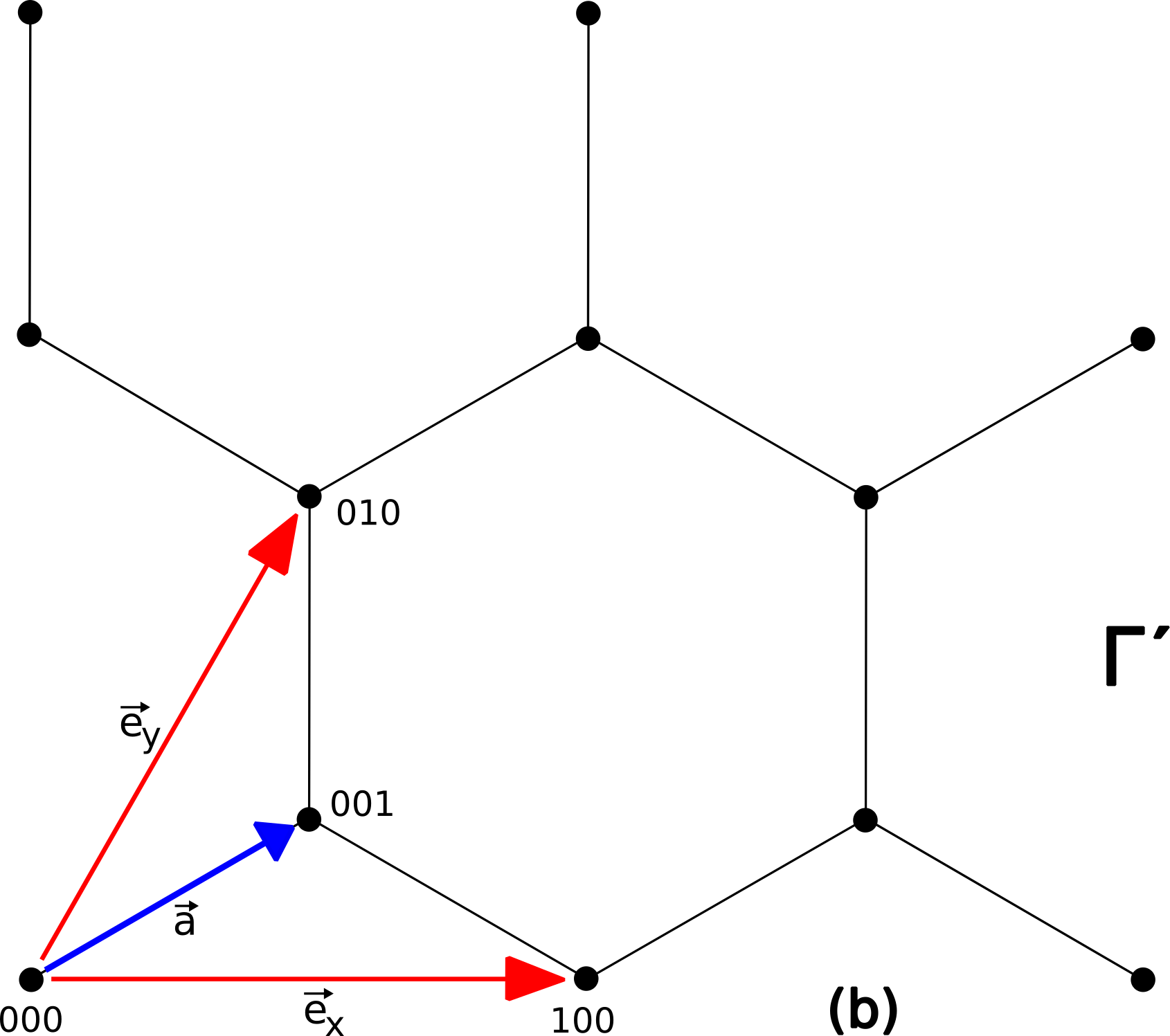}
\caption{Graph $\Gamma'$ is obtained from $\Gamma$ of Fig.~\ref{fig:honeycomb1} by replacing the blue polygons with single vertices. Vectors $\vec e_x$ and $\vec e_y$ are shown in red and the unit vector $\vec a$ in blue. We are taking $m=2$, because there are two hexagons in the $x$-direction and two hexagons in the $y$-direction. The boundary conditions are cyclic following the directions $\vec e_x$ and $\vec e_y$. The coined QW on $\Gamma'$ with the Grover coin is equivalent to the SQW on $\Gamma$ because they have the same evolution operator.}
\label{fig:honeycomb2}
\end{figure}

Vectors $\alpha$ associated with the blue polygons are 
\begin{eqnarray}\label{alpha_xyi}
\ket{\alpha_{x,y}^{(i)}}&=&\frac{1}{\sqrt 3}\left(\ket{x,y,i,0}+\ket{x,y,i,1}+
\right.\nonumber\\ &&\left.
\ket{x,y,i,2}\right),
\end{eqnarray}
and they split into two sets: When $i=0$, the polygon refers to a vertex that in the honeycomb has the form $(x,y,0)$, and  when $i=1$, the polygon refers to a vertex that in the honeycomb has the form $(x,y,1)$. The first set consists of inverted triangles in $\Gamma$ and the latter consists of usual triangles. Vectors $\ket{\alpha_{x,y}^{(i)}}$ need not necessarily be in uniform superposition and can assume different values on different vertices. When they are in uniform superposition, they produce the Grover coin at each vertex.

Vectors $\beta$ associated with the red polygons are
\begin{eqnarray}\label{beta_xyk}
\ket{\beta_{x,y}^{(k)}}&=&\frac{1}{\sqrt 2}\left(\ket{x,y,0,k}+
\right.\nonumber\\&&\left.
\ket{x-\delta_{1k},y-\delta_{2k},1,k}\right),
\end{eqnarray}
for $0\le k\le 2$. They must be in uniform superposition.

The evolution operator is $U=U_1 U_0$, where $U_0$ induces the blue tessellation and is given by
\begin{equation}
	U_0 \,=\, 2\sum_{x,y=0}^{m-1} 
	\left(\ket{\alpha_{xy}^{(0)}}\bra{\alpha_{xy}^{(0)}}+
	\ket{\alpha_{xy}^{(1)}}\bra{\alpha_{xy}^{(1)}}\right) - I,
\end{equation}
and $U_1$ induces the red tessellation and is given by
\begin{eqnarray}
	U_1 &=& 2\sum_{x,y=0}^{m-1} 
	\left(\ket{\beta_{xy}^{(0)}}\bra{\beta_{xy}^{(0)}}+
\ket{\beta_{xy}^{(1)}}\bra{\beta_{xy}^{(1)}}+\right.\nonumber\\
&&\left.\ket{\beta_{xy}^{(2)}}\bra{\beta_{xy}^{(2)}}\right) - I.
\end{eqnarray}
$U_1$ can be further simplified and reduced into the form of Eq.~(\ref{simpU_1}). $U$ is equal to the evolution operator of the coined model analyzed in Ref.~\cite{Abal:2010,LYW15}. In the continuous-time case, the honeycomb was analyzed in Refs.~\cite{FGT14,FGT15} using the evolution operator of the continuous-time QW model~\cite{Farhi:1998}.

There are three regular two-dimensional lattices: squared, triangular, and hexagonal. Only the coined model on the hexagonal lattice corresponds to SQWs on \textit{planar graphs}. Coined models on the squared and triangular lattice correspond to SQW on graphs that have 4-cliques and 6-cliques as \textit{induced graphs}, which are \textit{nonplanar}.

\subsection{Three-state coined QWs}

An interesting question is can we define a SQW on a \textit{directed graph}? The answer seems to be negative because if the walker is on vertex $v_1$ that has an edge pointing to $v_2$ and there is no edge from $v_2$ pointing to $v_1$ then there is a coin value $i$ so that $S\ket{v_1}\ket{i}=\ket{v_2}\ket{j}$ for some $j$ and there is no coin value $i$ so that  $S\ket{v_2}\ket{i}=\ket{v_1}\ket{j}$, where $S$ is the shift operator of a coined model on the directed graph. This means that $S^2\neq I$. Since the staggered model uses only Hermitian and unitary operators to produce the evolution operator, it seems that we cannot define SQWs on directed graphs that would be equivalent to the coined model. There is an exception when the directed edges are loops. In the next example, we show that a coined model on a directed graph is equivalent to a SQW on a graph that is in Class~2b$\setminus$Class~2b$'$.

\begin{figure}[ht!] 
\centering
\includegraphics[scale=0.68]{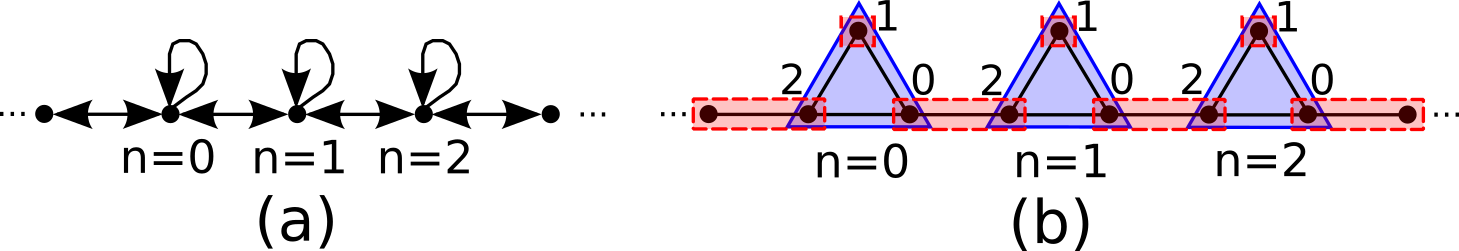}
\caption{Panel~(a) describes a directed graph on which a three-state coined QW is defined. Panel~(b) describes a 2-tessellable graph on which an equivalent SQW is defined.}
\label{fig:3states}
\end{figure}

Consider the flip-flop three-state coined QW defined on the directed graph of Fig.~\ref{fig:3states}(a), which was analyzed in Refs.~\cite{SBJ12,FB14,Mac15}. The evolution operator is $U=S\,(C(\rho)\otimes I)$, where $C(\rho)$ is the coin operator 
\begin{equation}
\label{coin_C}
C(\rho) = \left(
  \begin{array}{ccc}
    -\rho^2 & \rho\sqrt{2-2\rho^2} & 1-\rho^2 \\
    \rho\sqrt{2-2\rho^2} & 2\rho^2-1 & \rho\sqrt{2-2\rho^2} \\
    1-\rho^2 & \rho\sqrt{2-2\rho^2} & -\rho^2 \\
  \end{array}
\right),
\end{equation}
with parameter $\rho\in(0,1)$ and $S$ is the flip-flop shift operator
\begin{eqnarray}
S&=&\sum_{n\in \mathbb{Z}} \left(\ket{n+1,2}\bra{n,0}+\ket{n,1}\bra{n,1}+\right.\nonumber\\
 &&\left.\ket{n-1,0}\bra{n,2}\right).
\end{eqnarray}
The flip-flop shift operator is interesting because it does not use information that is external to the graph such as go to the right or go to the left and can be easily extended to generic graphs. The familiar three-state Grover walk~\cite{IKS05} is recovered taking $\rho=1/\sqrt 3$.

Coin $C(\rho)$ is an orthogonal reflection because it has only one $(+1)$-eigenvector, which is given by\begin{equation}
\ket{\alpha_\rho}\,=\,\sqrt{\frac{1-\rho^2}{2}}\ket{0}+\rho\ket{1}+\sqrt{\frac{1-\rho^2}{2}}\ket{2}.
\end{equation}
Then $C(\rho)=2\ket{\alpha_\rho}\bra{\alpha_\rho}-I$. Define 
\begin{equation}
U_0\,=\,2\sum_{n\in \mathbb{Z}} \ket{n,\alpha_\rho}\bra{n,\alpha_\rho}-I,
\end{equation}
where vectors $\ket{n,\alpha_\rho}$ induce the blue polygons of Fig.~\ref{fig:3states}(b) using labels $(n,i)$, $0\le i\le 2$, for the vertices of the graph. Notice that $U_0=C(\rho)\otimes I$. 

Define $U_1$ by
\begin{equation}
U_1\,=\,2\sum_{n\in \mathbb{Z}} \left(\ket{\beta_n^{(0)}}\bra{\beta_n^{(0)}}+\ket{\beta_n^{(1)}}\bra{\beta_n^{(1)}}\right)-I,
\end{equation} 
using the red tessellation, whose polygons are induced by
\begin{eqnarray}
\ket{\beta_n^{(0)}}&=&\frac{\ket{n,0}+\ket{n+1,2}}{\sqrt 2},\nonumber\\
\ket{\beta_n^{(1)}}&=&\ket{n,1}.
\end{eqnarray}
After simplifying $U_1$, we obtain $U_1=S$. Then, the SQW on the graph of Fig.~\ref{fig:3states}(b) with evolution operator $U=U_1 U_0$ is equivalent to the flip-flop coined QW on the directed graph of Fig.~\ref{fig:3states}(a). The blue tessellation represents the internal coin states and the red tessellation represents the shift operator.

\section{Searching}\label{sec6}

Searching in the staggered model on a graph $\Gamma$ is implemented using partial tessellations. A partial tessellation is a tessellation of an induced subgraph of $\Gamma$ (distinct from $\Gamma$). The vertices in the missing polygons are the marked ones. If the evolution operator is $U=U_1U_0$ and both $U_0$ and $U_1$ are unitary operators associated with partial tessellations, then we demand that the tessellation union covers all vertices. The vertices that belong to only one polygon are the marked ones and the vertices that belong to two polygons are the ordinary ones. If $U_0$ induces a partial tessellation and $U_1$ induces a complete tessellation, the tessellation union always covers all vertices of the graph because the tessellation induced by $U_1$ covers all vertices. A SWQ using at least one partial tessellation is called a \textit{generalized} SQW.

The initial condition is the uniform superposition of all vertices in order to avoid any bias towards the location of the marked vertices, that is
\begin{equation}
	\ket{\psi_0}\,=\,\frac{1}{\sqrt{|V|}}\sum_{v\in V} \ket{v},
\end{equation}
where $V$ is the vertex set of $\Gamma$. The searching algorithm consists of applying $U$ in succession, that is, the final state is $\ket{\psi_t}=U^t\ket{\psi_0}$, where $t$ is the running time.

\begin{figure}[ht!] 
\centering
\includegraphics[scale=0.42]{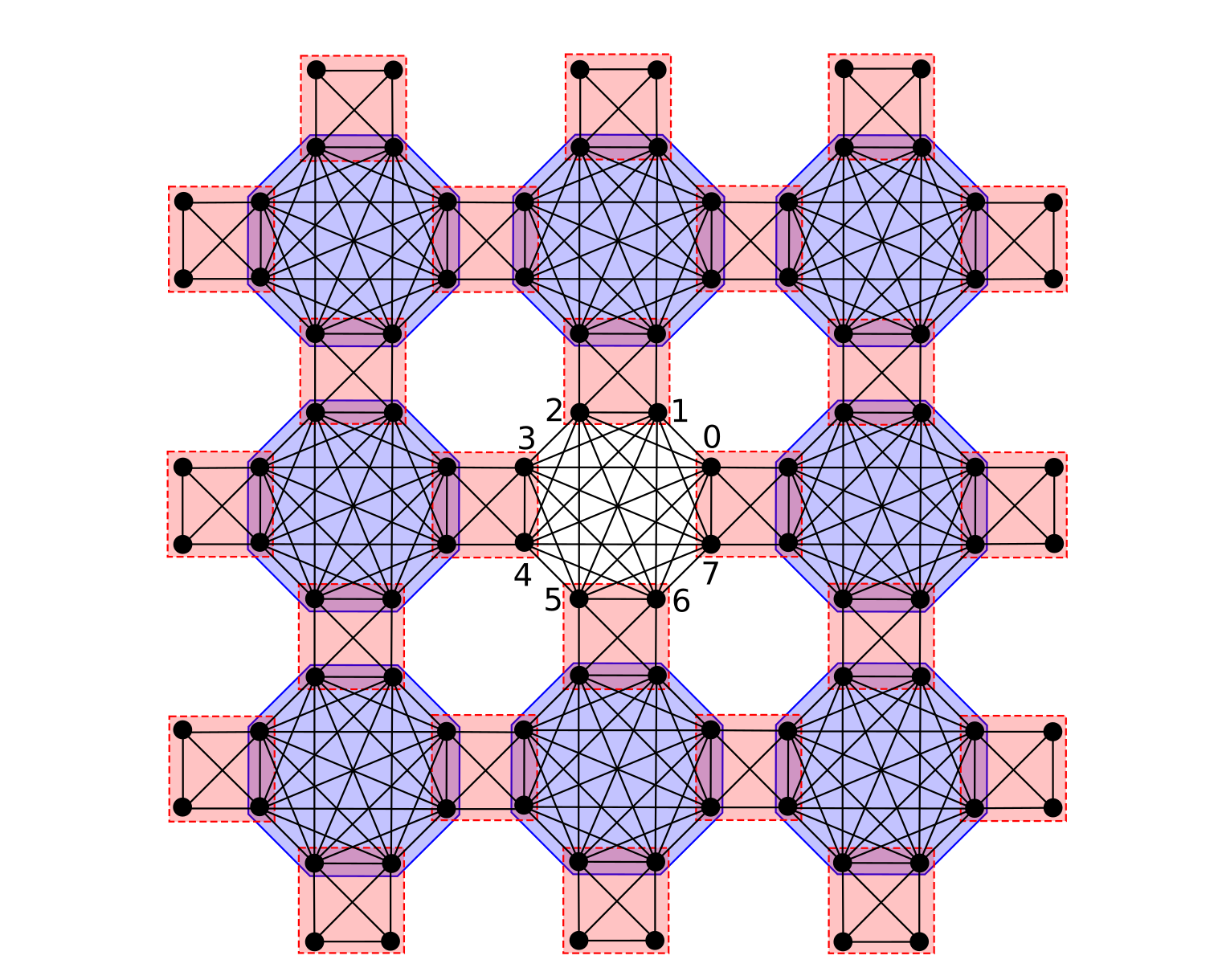}
\caption{Graph in Class~1 with blue 8-clique partial tessellation and red 4-clique complete tessellation. The vertices in the central 8-clique are the marked ones. The 4-cliques at the boundary are identified establishing a torus-like topology.}
\label{fig:graph6}
\end{figure}

We present an example of a generalized SQW on a graph in Class~1 that finds a marked vertex faster than classical random walks. Let us define a generalized SQW on the graph of Fig.~\ref{fig:graph6}. This example is included neither in the coined nor in the Szegedy model because there are edges in the tessellation intersection. This graph is in Class~1 because graph $\Gamma_4$ of Fig.~\ref{fig:app4}, for instance, is an \textit{induced subgraph}. The graph consists of $n^2$ 8-cliques linked by $2n^2$ 4-cliques with a torus-like topology, which is obtained by identifying the external 4-cliques. This graph is 2-tessellable as can be checked in Fig.~\ref{fig:graph6}, which depicts the case $n=3$. The blue polygons cover the 8-cliques and the red polygons cover the 4-cliques. There is a missing blue polygon associated with the central 8-clique characterizing a partial blue tessellation. The marked vertices are the ones in the central 8-clique. All vertices are in the tessellation union. Notice that there are edges that do not belong to the tessellation union. This can happen in the generalized model.

The Hilbert space associated with this graph has dimension $N=8 n^2$. The vectors associated with the blue polygons are
\begin{equation}
\ket{\alpha_{x y}}\,=\, \frac{1}{2\sqrt 2}\sum_{k=0}^7 \ket{x,y,k},
\end{equation}
and the vectors associated with the red polygons are
\begin{eqnarray}
\ket{\beta_{x y}^{(0)}} &=& \frac{1}{2}\,\ket{x,y}\left(\ket{0}+\ket{7}\right)+\nonumber\\
&& \frac{1}{2}\,\ket{x+1,y}\left(\ket{3}+\ket{4}\right),\\
\ket{\beta_{x y}^{(1)}} &=& \frac{1}{2}\,\ket{x,y}\left(\ket{1}+\ket{2}\right)+\nonumber\\
&& \frac{1}{2}\,\ket{x,y+1}\left(\ket{5}+\ket{6}\right),
\end{eqnarray}
for $0\le x,y\le n-1$ and the arithmetic with the labels of $\ket{x,y}$ is performed modulo $n$. The central 8-clique is located at $x=0,y=0$ (with no blue polygon) and the vertices have labels $(x,y,k)$, where $k$ runs from 0 to 7, as shown in Fig.~\ref{fig:graph6}.

The evolution operator is $U=U_1 U_0$, where $U_0$ induces the blue tessellation, given by
\begin{equation}
	U_0 \,=\, 2\sum_{\mathclap{\substack{x,y=0\\
                   (x,y)\neq(0,0)}}}^{n-1} \ket{\alpha_{xy}}\bra{\alpha_{xy}} - I,
\end{equation}
and $U_1$ induces the red tessellation, given by
\begin{equation}
	U_1 \,=\, 2\sum_{x,y=0}^{n-1}\left(\ket{\beta_{xy}^{(0)}}\bra{\beta_{xy}^{(0)}}+
	\ket{\beta_{xy}^{(1)}}\bra{\beta_{xy}^{(1)}}\right) - I.
\end{equation}

\begin{figure}[ht!] 
\centering
\includegraphics[trim=50 490 250 60,clip,scale=0.5]{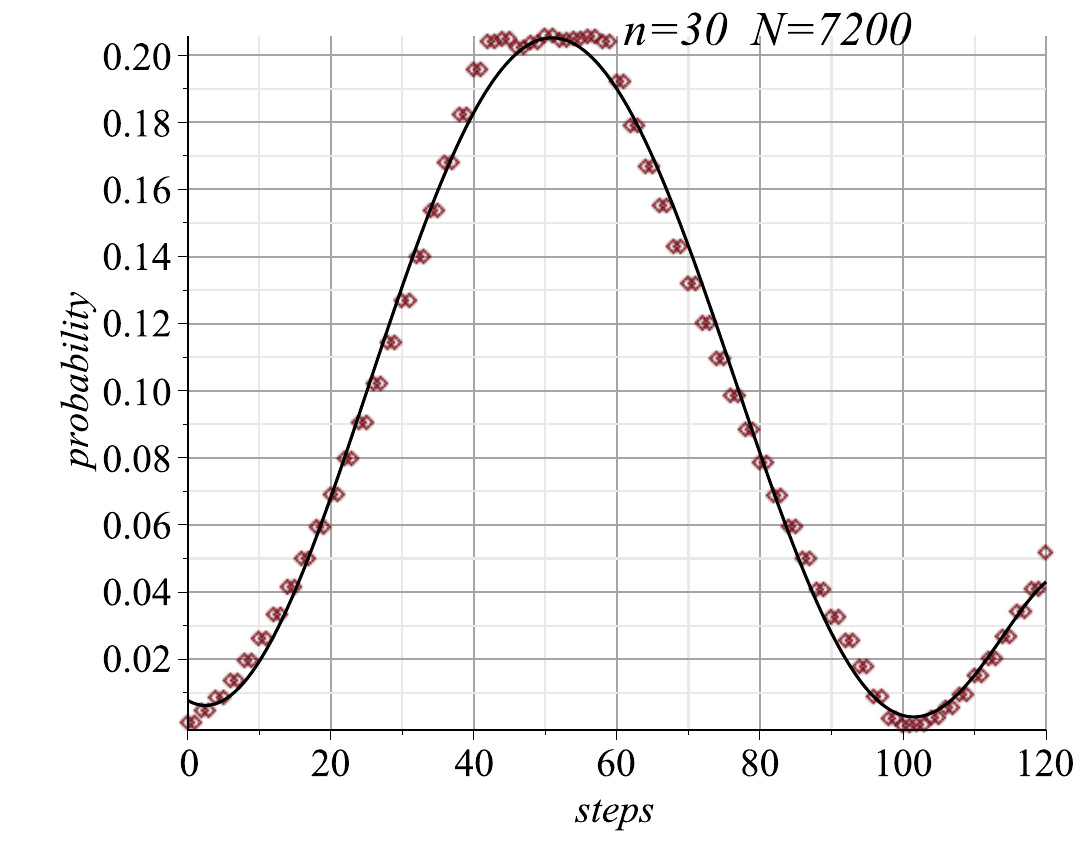}\\
\includegraphics[trim=50 490 300 60,clip,scale=0.59]{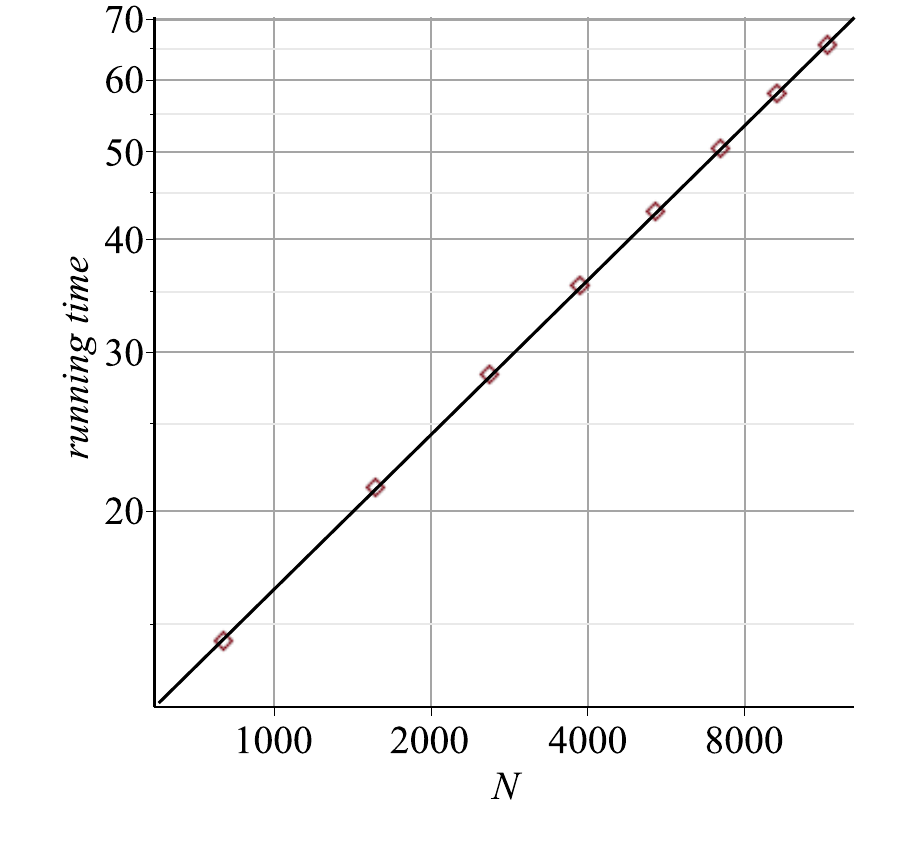}\\
\includegraphics[trim=50 490 300 60,clip,scale=0.6]{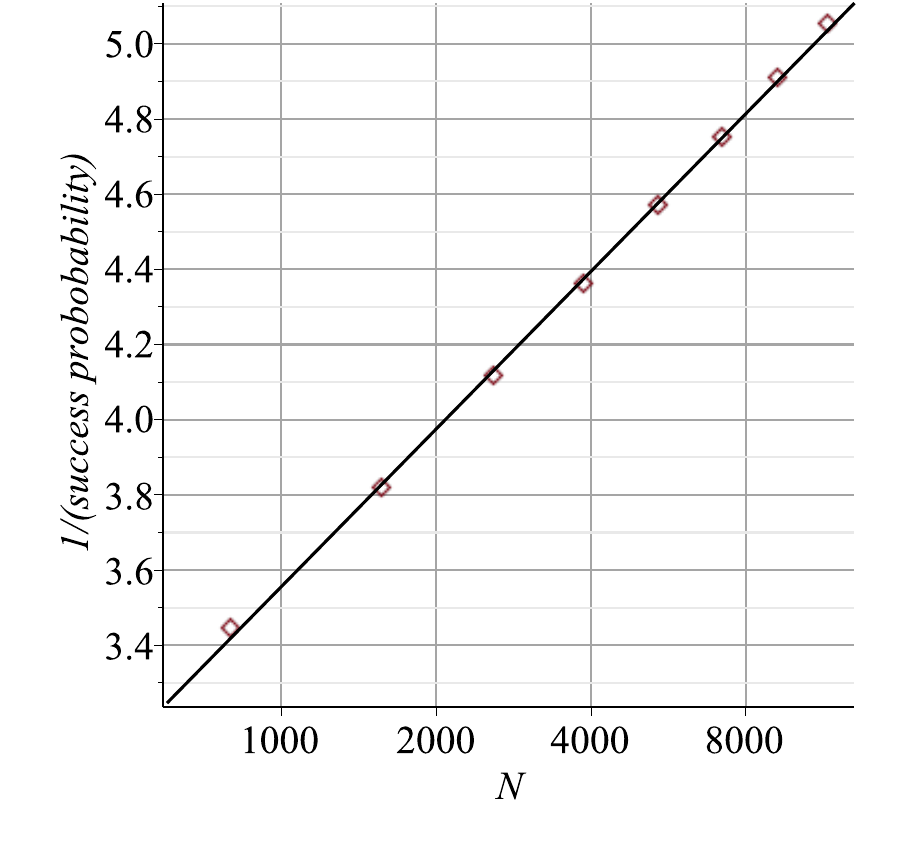}
\caption{Three plots that allow us to determine the time complexity of the searching algorithm. The first panel depicts the probability of finding the walker in a marked vertex as a function of the number of steps. The second panel depicts the loglog plot of the running time as function of the number of vertices $N$. The third panel depicts the semilog plot of the inverse of the success probability as function of $N$. In all of them we included the fitting curve using the least-square method.}
\label{fig:prob_vs_steps_n=20}
\end{figure}

To perform a numerical analysis to obtain the efficiency of the search algorithm based on $U$, we use the following three-part procedure. Firstly, we fix a value of $n$ and we plot the probability of finding the walker in a marked vertex as a function of the number of steps. This probability is given by
\begin{equation}
	p(t)\,=\, \sum_{k=0}^7 \bra{0,0,k}U^t\ket{\psi_0},
\end{equation}
where $t$ is the number of steps and $\psi_0$ is the initial condition, given by the uniform superposition. The first panel of Fig.~\ref{fig:prob_vs_steps_n=20} shows $p(t)$ for $n=30$ and the fitting curve. We take the first maximum of the fitting curve as the success probability and the corresponding abscissa as the running time of the searching algorithm. This process is repeated for many values of $n$. Secondly, we analyze how the running time increases as a function of the number of vertices $N$, which is depicted in the second panel of Fig.~\ref{fig:prob_vs_steps_n=20}. The fitting line shows that the running time is approximately $0.32 N^{0.57}$. Thirdly, we analyze how the success probability decreases as a function of $N$, which is depicted in the third panel of Fig.~\ref{fig:prob_vs_steps_n=20}. The fitting line shows that the success probability is approximately $0.53/(\ln N)^{0.60}$. Joining those three steps and extrapolating for large $N$, we conclude that the total running time with a constant success probability is $O( N^{0.57}(\ln N)^{0.30})$ because we use the amplitude amplification scheme~\cite{BBHT98}.

\section{Discussions and conclusions}\label{sec7}

An interesting question in this context is: Is there some advantage of using the SQW model? A partial answer comes from the following four points: (1)~When coined QWs are converted into the SQW model, it becomes clear that the coin and shift operators have the same nature and the distinction between them vanishes. What matters is the alternate action of two unitary operators. For example, decoherence by breaking edges of the graph, called percolation~\cite{Romanelli:2005,Oliveira:2006,KKNJ12}, can be applied to the edges that represent the coin. (2)~The SQW model helps to establish the equivalence between the coined and Szegedy's model~\cite{Por16}. (3)~The SQW model is an important step to unify the description of discrete-time QW models. (4)~The SQW model may help to understand the connection between coined and continuous-time QWs~\cite{Farhi:1998}. Ref.~\cite{DB15} described how to define a coined QW with Hermitian and unitary coins with a well-defined continuous-time limit. The authors used a method of expanding the graph on which the coined model takes place. This method is similar to the technique employed in Sec.~\ref{sec5} to express the coined QW as an instance of a staggered QW.

In this paper we have characterized which graphs are 2-tessellable by proving that a graph is 2-tessellable if and only if its clique graph is bipartite. The class of 2-tessellable SQWs is large enough to encompass Szegedy's model. Since it is defined using the product of only two local operators, we can employ Szegedy's spectral lemma~\cite{Szegedy:2004,PSFG15} to find the spectral decomposition of the evolution operator. We have also shown that 2-tessellable SQWs with no edge in the tessellation intersection can be cast into Szegedy's model.

Another contribution of this paper is to characterize the classes of graphs that help to establish the equivalence among discrete-time QW models. We use four classes:
\begin{itemize}
\item Class~1 -- graphs that are not line graphs.
\item Class~2a -- graphs that are line graphs of nonbipartite graphs.
\item Class~2b -- graphs that are line graphs of bipartite graphs.
\item Class~2b$'$ -- graphs that obey conditions (1) to (3) described on page~\pageref{3conditions}.
\end{itemize} 
We have shown that if a SQW with two tessellations is defined in Class~2b$'$ then it can be reduced to the coined model and any SQW defined in Classes 1 and 2a are not equivalent to Szegedy's QW. Besides, any SQW that can be cast into Szegedy's model must be in Class~2b, but there are SQWs in Class~2b that cannot be cast into Szegedy's model. Ref.~\cite{PSFG15} showed that if we convert an extended Szegedy's QW on $\Gamma$ into an equivalent SQW on the line graph of $\Gamma$, there is no edge in the tessellation intersection. Proposition~\ref{mainproposition} formally shows the inverse, that is, if there is exactly one vertex in each polygon intersection, then the SQW can be cast into the extended Szegedy model.

We have provided two examples of nontrivial 2-tessellable SQWs that are equivalent to the coined model, which help to understand the connection between the coined and staggered models. We have also given an example of a searching algorithm using a SQW on a graph in Class~1, which cannot be reduced to Szegedy's model. A numerical analysis has shown that this algorithm is more efficient than its classical analogue using random walks.

\appendix

\numberwithin{equation}{section}



\section*{Appendices}

\section{Glossary of terms in graph theory}\label{appendixA}

This appendix compiles the main definitions of graph theory used in this work~\cite{Die12,BLS99,Har94}.

A \textit{simple undirected graph} $\Gamma(V,E)$ is defined by a set $V$ of vertices or nodes and a set $E$ of edges so that each edge links two vertices and two vertices are linked by at most one edge. Two vertices linked by an edge are called \textit{adjacent}. Two edges  that share a common vertex are also called adjacent. The \textit{degree} of a vertex is the number of edges incident to the vertex. A graph is \textit{connected} when there is a path between every pair of vertices, otherwise it is called \textit{disconnected}. The \textit{complete graph} is a simple graph in which every pair of distinct vertices is connected by an edge.  A \textit{directed graph} is a graph whose edges have a direction associated with them. A \textit{multigraph} is an extension of the definition of graph that allows multiple edges between vertices. Most of the times, we use the term \textit{graph} as synonym of \textit{simple undirected graph}.

A subgraph $\Gamma'(V',E')$, where $V'\subset V$ and $E'\subset E$, is an \textit{induced subgraph} of $\Gamma(V,E)$ if it has exactly the edges that appear in $\Gamma$ over the same vertex set. If two vertices are adjacent in $\Gamma$ they are also adjacent in the induced subgraph. 

A \textit{bipartite graph} is a graph whose vertex set $V$ is the union of two disjoint sets $X$ and $X'$ so that no two vertices in $X$ are adjacent and no two vertices in $X'$ are adjacent. A \textit{complete bipartite graph} is a bipartite graph such that every possible edge that could connect vertices in $X$ and $X'$ is part of the graph and is denoted by $K_{m,n}$, where $m$ and $n$ are the cardinalities of sets $X$ and $X'$, respectively.

A \textit{clique} is a subset of vertices of a graph such that its induced subgraph is complete. A \textit{maximal clique} is a clique that cannot be extended by including one more adjacent vertex, that is, it is not contained in a larger clique. A \textit{maximum clique} is a clique of maximum possible size. A clique of size $d$ is called a $d$-\textit{clique}. A clique can have one vertex. Some references in graph theory use the term ``clique'' as synonym of \textit{maximal clique}. We avoid this notation here.

A \textit{clique graph} $K(\Gamma)$ of a graph $\Gamma$ is a graph such that every vertex represents a maximal clique of $\Gamma$ and two vertices of $K(\Gamma)$ are adjacent if and only if the underlying maximal cliques in $\Gamma$ share at least one vertex in common.

A \textit{clique partition} of a graph $\Gamma$ is a set of cliques of $\Gamma$ that contains each edge of $\Gamma$ exactly once. A \textit{minimum clique partition} is a clique partition with the smallest set of cliques. A \textit{clique cover} of a graph $\Gamma$ is a set of cliques of $\Gamma$ that contains each edge of $\Gamma$ at least once.  A \textit{minimum clique cover} is a clique cover with the smallest set of cliques.

A \textit{diamond graph} is a graph with 4 vertices and 5 edges consisting of a 4-clique minus one edge or two triangles sharing a common edge. A graph is \textit{diamond-free} if no induced subgraph is isomorphic to a \textit{diamond graph}.

A \textit{line graph} (or \textit{derived graph} or \textit{interchange graph}) of a graph $\Gamma$ (called \textit{root graph}) is another graph $L(\Gamma)$ so that each vertex of $L(\Gamma)$ represents an edge of $\Gamma$ and two vertices of $L(\Gamma)$ are adjacent if and only if their corresponding edges share a common vertex in $\Gamma$. 

A \textit{matching} $M\subseteq E$ is a set of edges without pairwise common vertices. An edge $m\in M$ \textit{matches} the endpoints of $m$. A \textit{perfect matching} is a matching that matches all vertices of the graph.

A \textit{planar graph} is a graph that can be drawn in a two-dimensional plane in such a way that no edges cross each other.

A \textit{proper coloring} or simply \textit{coloring} of a loopless graph is a labeling of the vertices with colors such that no two vertices sharing the same edge have the same color. A $k$-\textit{colorable} graph is the one whose vertices can be colored with at most $k$ colors so that no two adjacent vertices share the same color. This concept can be used for edges and other graph structures.

\section{Definition of Szegedy's QW}\label{appendixB}

Let us define Szegedy's QW model~\cite{Szegedy:2004} using the description given in Ref.~\cite{Por16}. Consider a connected bipartite graph $\Gamma(X,Y,E)$, where $X,Y$ are disjoint sets of vertices and $E$ is the set of non-directed edges. Let 
\begin{equation}\label{biadmatrix}
	\left(
	\begin{array}{cc}
		  0 & A \\
		  A^T & 0
	\end{array}
	\right)
\end{equation}
be the biadjacency matrix of $\Gamma(X,Y,E)$. Using $A$, define $P$ as a probabilistic map from $X$ to $Y$ with entries $p_{xy}$. Using $A^T$, define $Q$  as a probabilistic map from $Y$ to $X$ with entries $q_{yx}$. If $P$ is an $m\times n$ matrix, $Q$ will be an $n\times m$ matrix. Both are right-stochastic, that is, each row sums to 1. Using $P$ and $Q$, it is possible to define unit vectors
\begin{eqnarray}
  \ket{\phi_x} &=&  \sum_{y\in Y} \sqrt{p_{x y}}\,\textrm{e}^{i\theta_{xy}} \, \ket{x,y}, \label{ht_phi_x} \\
  \ket{\psi_y}  &=&  \sum_{x\in X} \sqrt{q_{y x}}\,\textrm{e}^{i\theta'_{xy}} \, \ket{x,y}, \label{ht_psi_y}
\end{eqnarray}
that have the following properties: $\braket{\phi_x}{\phi_{x'}}=\delta_{xx'}$ and $\braket{\psi_y}{\psi_{y'}}=\delta_{yy'}$. In Szegedy's original definition, $\theta_{xy}=\theta'_{xy}=0$. We call \textbf{extended Szegedy's QW} the version that allows nonzero angles.

\begin{definition}\label{def:SzegedyQW}
\textbf{Szegedy's QW} on a bipartite graph $\Gamma(X,Y,E)$ with biadjacent matrix (\ref{biadmatrix}) is defined on a Hilbert space ${\cal H}^{m n} = {\cal H}^{m}\otimes {\cal H}^{n} $, where $ m = | X |$ and $n = | Y | $, the computational basis of which is $ \big \{\ket {x, y}: x \in X, y \in Y \big \} $.
The QW is driven by the unitary operator
\begin{equation}\label{ht_U_ev}
    W \,=\, R_1 \, R_0,
\end{equation}
where
\begin{eqnarray}
  R_0 &=& 2\sum_{x\in X} \ket{\phi_x}\bra{\phi_x} - I, \label{ht_RA}\\
  R_1 &=& 2\sum_{y\in Y} \ket{\psi_y}\bra{\psi_y} - I. \label{ht_RB}
\end{eqnarray}
\end{definition}
Notice that operators $R_0$ and $R_1$ are unitary and Hermitian ($R_0^2=R_1^2=I$).

\section*{Acknowledgements}
The author acknowledges financial support from Faperj (grant n.~E-26/102.350/2013) and CNPq (grants n.~303406/2015-1, 4741\-43/2013-9). The author also acknowledges useful discussions with Franklin Marquezino and Luerbio Faria.


\begin{thebibliography}{10}

\bibitem{Aharonov:2000}
D. Aharonov, A. Ambainis, J. Kempe, and U. Vazirani.
\newblock Quantum walks on graphs.
\newblock In {\em Proceedings of the Thirty-third Annual ACM Symposium on
  Theory of Computing}, STOC '01, pages 50--59, New York, USA, 2001. 

\bibitem{Ven12}
S.~E. Venegas-Andraca.
\newblock Quantum walks: a comprehensive review.
\newblock {\em Quantum Information Processing}, 11(5):1015--1106, 2012.

\bibitem{Kon08}
Norio Konno.
\newblock Quantum walks.
\newblock In U.~Franz and M.~Schuermann, editors, {\em Quantum Potential
  Theory}, volume 1954 of {\em Lecture Notes in Mathematics}, pages 309--452.
  Springer, Berlin, 2008.

\bibitem{Kendon:2007}
Viv Kendon.
\newblock Decoherence in quantum walks - a review.
\newblock {\em Mathematical Structures in Computer Science}, 17(6):1169--1220,
  2007.

\bibitem{Portugal:Book}
Renato Portugal.
\newblock {\em Quantum Walks and Search Algorithms}.
\newblock Springer, New York, 2013.

\bibitem{Manouchehri2014}
K. Manouchehri and J. Wang.
\newblock {\em {Physical Implementation of Quantum Walks}}.
\newblock Springer, Berlin, 2014.

\bibitem{Shenvi:2003}
N.~Shenvi, J.~Kempe, and K.~B. Whaley.
\newblock Quantum random-walk search algorithm.
\newblock {\em Phys. Rev. A}, 67:052307, 2003.

\bibitem{Ambainis:2005}
A.~Ambainis, J.~Kempe, and A.~Rivosh.
\newblock Coins make quantum walks faster.
\newblock In {\em Proceedings of the 16th ACM-SIAM Symposium on Discrete
  Algorithms}, pages 1099--1108, 2005.

\bibitem{Szegedy:2004}
M.~Szegedy.
\newblock Quantum speed-up of {M}arkov chain based algorithms.
\newblock In {\em Proceedings of the 45th Symposium on Foundations of Computer
  Science}, pages 32--41, 2004.

\bibitem{mss07}
F.~Magniez, M.~Santha, and M.~Szegedy.
\newblock Quantum algorithms for the triangle problem.
\newblock {\em SIAM Journal on Computing}, 37(2):413--424, 2007.

\bibitem{PMCM13}
G.~D. Paparo, M.~M\"{u}ller, F.~Comellas, and M.~A. Martin-Delgado.
\newblock {Quantum Google in a Complex Network}.
\newblock {\em Scientific Reports}, 3:2773, 2013.

\bibitem{HKSS14}
Y.~Higuchi, N.~Konno, I.~Sato, and E.~Segawa.
\newblock Spectral and asymptotic properties of grover walks on crystal
  lattices.
\newblock {\em Journal of Functional Analysis}, 267(11):4197 -- 4235, 2014.

\bibitem{MOS16}
K.~Matsue, O.~Ogurisu, and E.~Segawa.
\newblock Quantum walks on simplicial complexes.
\newblock {\em Quantum Information Processing}, 15(5):1865--1896, 2016.

\bibitem{PSFG15}
R.~Portugal, R.~A.~M. Santos, T.~D. Fernandes, and D.~N. Gon{\c{c}}alves.
\newblock The staggered quantum walk model.
\newblock {\em Quantum Information Processing}, 15(1):85--101, 2016.

\bibitem{Por16}
Renato Portugal.
\newblock Establishing the equivalence between {S}zegedy's and coined quantum
  walks using the staggered model.
\newblock {\em Quantum Information Processing}, 15(4):1387--1409, 2016.

\bibitem{KS75}
J. Kogut and L. Susskind.
\newblock Hamiltonian formulation of {W}ilson's lattice gauge theories.
\newblock {\em Phys. Rev. D}, 11:395--408, 1975.

\bibitem{Sus77}
Leonard Susskind.
\newblock Lattice fermions.
\newblock {\em Phys. Rev. D}, 16:3031--3039, 1977.

\bibitem{STW81}
H.~S. Sharatchandra, H.~J. Thun, and P.~Weisz.
\newblock Susskind fermions on a euclidean lattice.
\newblock {\em Nuclear Physics B}, 192(1):205 -- 236, 1981.

\bibitem{Kog79}
John~B. Kogut.
\newblock An introduction to lattice gauge theory and spin systems.
\newblock {\em Rev. Mod. Phys.}, 51:659--713, 1979.

\bibitem{Meyer96}
David~A. Meyer.
\newblock From quantum cellular automata to quantum lattice gases.
\newblock {\em Journal of Statistical Physics}, 85(5-6):551--574, 1996.

\bibitem{Mey96b}
David~A. Meyer.
\newblock On the absence of homogeneous scalar unitary cellular automata.
\newblock {\em Physics Letters A}, 223(5):337 -- 340, 1996.

\bibitem{Patel05}
A.~Patel, K.~S. Raghunathan, and P.~Rungta.
\newblock Quantum random walks do not need a coin toss.
\newblock {\em Phys. Rev. A}, 71:032347, 2005.

\bibitem{Patel:2010}
A.~Patel, K.~S. Raghunathan, and Md.A. Rahaman.
\newblock Search on a hypercubic lattice using a quantum random walk. ii.  $d=2$.
\newblock {\em Phys. Rev. A}, 82:032331, 2010.

\bibitem{Patel:2010b}
A.~Patel and Md.~A. Rahaman.
\newblock Search on a hypercubic lattice using a quantum random walk. i. $d>2$.
\newblock {\em Phys. Rev. A}, 82:032330, 2010.

\bibitem{HKS05}
M.~Hamada, N.~Konno, and E.~Segawa.
\newblock Relation between coined quantum walks and quantum cellular automata.
\newblock {\em RIMS Kokyuroku}, 1422:1--11, 2005.

\bibitem{Note1}
In this work, we employ many technical terms of graph theory, which are in
  italics to indicate that they are in the glossary in Appendix~\ref
  {appendixA}.

\bibitem{KMOR15}
H.~Krovi, F.~Magniez, M.~Ozols, and J.~Roland.
\newblock Quantum walks can find a marked element on any graph.
\newblock {\em Algorithmica (on-line)}, pages 1--57, 2015.

\bibitem{Ben02}
Paul Benioff.
\newblock Space searches with a quantum robot.
\newblock {\em AMS Contemporary Math Series}, 305, 2002.

\bibitem{Ambainis:2013}
A.~Ambainis, R.~Portugal, and N.~Nahimov.
\newblock Spatial search on grids with minimum memory.
\newblock {\em Quantum Information \& Computation}, 15:1233--1247, 2015.

\bibitem{Note2}
We use the following convention in this work: Polygons of tessellation $\alpha$ 
are surrounded by continuous boundaries and filled with transparent blue
  color. Polygons of tessellation $\beta $ are surrounded by dashed boundaries
  and filled with transparent red color. The color choices play no relevant
  role.

\bibitem{Szw03}
J.~L. Szwarcfiter.
\newblock {\em Recent Advances in Algorithms and Combinatorics}, chapter A
  Survey on Clique Graphs, pages 109--136.
\newblock Springer, New York, 2003.

\bibitem{Bei70}
L.~W. Beineke.
\newblock {Characterizations of derived graphs}.
\newblock {\em Journal of Combinatorial Theory}, 9(2):129--135, 1970.

\bibitem{Kra43}
J.~Krausz.
\newblock D\'emonstration nouvelle d'une th\'eor\`eme de {W}hitney sur les
  r\'eseaux.
\newblock {\em Mat. Fiz. Lapok}, 50:75--85, 1943.

\bibitem{Pet03}
Dale Peterson.
\newblock Gridline graphs: a review in two dimensions and an extension to
  higher dimensions.
\newblock {\em Discrete Applied Mathematics}, 126(2–3):223 -- 239, 2003.

\bibitem{Abal:2010}
G.~Abal, R.~Donangelo, F.~L. Marquezino, and R.~Portugal.
\newblock Spatial search on a honeycomb network.
\newblock {\em Mathematical Structures in Computer Science}, 20:999--1009,
  2010.

\bibitem{LYW15}
C. Lyu, L. Yu, and S. Wu.
\newblock Localization in quantum walks on a honeycomb network.
\newblock {\em Phys. Rev. A}, 92:052305, 2015.

\bibitem{FGT14}
I. Foulger, S. Gnutzmann, and G. Tanner.
\newblock Quantum search on graphene lattices.
\newblock {\em Phys. Rev. Lett.}, 112:070504, 2014.

\bibitem{FGT15}
I. Foulger, S. Gnutzmann, and G. Tanner.
\newblock Quantum walks and quantum search on graphene lattices.
\newblock {\em Phys. Rev. A}, 91:062323, 2015.

\bibitem{Farhi:1998}
E.~Farhi and S.~Gutmann.
\newblock Quantum computation and decision trees.
\newblock {\em Physical Review A}, 58:915--928, 1998.

\bibitem{SBJ12}
M.~Stefanak, I.~Bezdekova, and I.~Jex.
\newblock Continuous deformations of the grover walk preserving localization.
\newblock {\em The European Physical Journal D}, 66:142, 2012.

\bibitem{FB14}
S. Falkner and S. Boettcher.
\newblock Weak limit of the three-state quantum walk on the line.
\newblock {\em Phys. Rev. A}, 90:012307, 2014.

\bibitem{Mac15}
Takuya Machida.
\newblock Limit theorems of a 3-state quantum walk and its application for
  discrete uniform measures.
\newblock {\em Quantum Information {\&} Computation}, 15(5{\&}6):406--418,
  2015.

\bibitem{IKS05}
N. Inui, N. Konno, and E. Segawa.
\newblock One-dimensional three-state quantum walk.
\newblock {\em Phys. Rev. E}, 72:056112, 2005.

\bibitem{BBHT98}
Michel Boyer, Gilles Brassard, Peter H{\o}yer, and Alain Tapp.
\newblock Tight bounds on quantum searching.
\newblock {\em Forstschritte Der Physik}, 4:820--831, 1998.

\bibitem{Romanelli:2005}
A.~Romanelli, R.~Siri, G.~Abal, A.~Auyuanet, and R.~Donangelo.
\newblock Decoherence in the quantum walk on the line.
\newblock {\em Physica A}, 347(C):137--152, 2005.

\bibitem{Oliveira:2006}
A.~C. Oliveira, R.~Portugal, and R.~Donangelo.
\newblock Decoherence in two-dimensional quantum walks.
\newblock {\em Physical Review A}, 74:012312, 2006.

\bibitem{KKNJ12}
B.~Koll\'ar, T.~Kiss, J.~Novotn\'y, and I.~Jex.
\newblock Asymptotic dynamics of coined quantum walks on percolation graphs.
\newblock {\em Phys. Rev. Lett.}, 108:230505, 2012.

\bibitem{DB15}
Dheeraj M~N and T.~A. Brun.
\newblock Continuous limit of discrete quantum walks.
\newblock {\em Phys. Rev. A}, 91:062304, 2015.

\bibitem{Die12}
Reinhard Diestel.
\newblock {\em Graph Theory}, volume 173 of {\em Graduate texts in
  mathematics}.
\newblock Springer, 2012.

\bibitem{BLS99}
A. Brandst\"{a}dt, V.~B. Le, and J.~P. Spinrad.
\newblock {\em Graph Classes: A Survey}.
\newblock Society for Industrial and Applied Mathematics, Philadelphia, PA,
  USA, 1999.

\bibitem{Har94}
Frank Harary.
\newblock {\em Graph Theory}.
\newblock Addison-Wesley Series in Mathematics. Perseus Books, 1994.

\end{thebibliography}

\end{document}